\documentclass[11pt]{article}

\pdfoutput=1
\usepackage{arxiv}
\usepackage{hyperref}
\usepackage[normalem]{ulem}
\usepackage{amsthm}
\usepackage{amsmath}
\usepackage{amssymb}
\usepackage{mathrsfs}
\usepackage{mathtools}
\usepackage{graphicx}
\usepackage{enumerate}
\usepackage{comment}
\usepackage{algorithm}
\usepackage{algpseudocode}
\usepackage{color}
\usepackage{stmaryrd}
\usepackage{booktabs} 
\usepackage{fullpage}
\usepackage{appendix}

\newtheorem{theorem}{Theorem}

\newtheorem{lemma}[theorem]{Lemma}

\theoremstyle{definition}
 
\theoremstyle{remark}
\newtheorem{remark}[theorem]{Remark}
\theoremstyle{remark}
\newtheorem{example}[theorem]{Example}
\numberwithin{theorem}{section}

\providecommand{\N}{}

\renewcommand{\N}{{\mathbb N}}

\newcommand{\E}[1]{{\mathbb E}\left[#1\right]}

\newcommand{\p}[1]{\mathbb{P}\left(#1\right)}

\newcommand{\Cprob}[2]{\mathbb{P}\left(\left. #1 \; \right| \; #2\right)}

\newcommand\cC{\mathcal C}

\newcommand\cU{{\mathcal U}}

\newcommand{\ER}{Erd\H{o}s-R\'e{}nyi }

\newcommand{\pairs}[1]{\ensuremath{\mathrm{sp}\left(#1\right)}}
\newcommand{\fixedsra}{$\vrule height 2\fontdimen22\textfont2 width 0pt\shortrightarrow$}
\newcommand{\shortarrow}[1]{%
  \mathrel{\text{\rotatebox[origin=c]{\numexpr#1*45}{\fixedsra}}}
}
\newcommand{\uppairs}[1]{\ensuremath{\mathrm{sp}^{\!\shortarrow{1}}\!\left(#1\right)}}
\newcommand{\downpairs}[1]{\ensuremath{\mathrm{sp}^{\!\shortarrow{7}}\!\left(#1\right)}}
\newcommand{\upmeasure}[1]{\ensuremath{\sigma^{\shortarrow{1}}_{\mathrm{SR}}\left(#1\right)}}
\newcommand{\downmeasure}[1]{\ensuremath{\sigma^{\shortarrow{7}}_{\mathrm{SR}}\left(#1\right)}}
\newcommand{\measure}[1]{\ensuremath{\sigma_{\mathrm{SR}}\left(#1\right)}}
\newcommand{\measurematrix}[1]{\ensuremath{M_{\mathrm{SR}}\left(#1\right)}}
\newcommand{\timematrix}[1]{\ensuremath{M^{\rightarrow}_{\mathrm{SR}}\left(#1\right)}}
\newcommand{\measurename}{simplicial ratio}
\newcommand{\matrixname}{simplicial matrix}
\newcommand{\upname}{bottom-up simplicial ratio}
\newcommand{\downname}{top-down simplicial ratio}
\newcommand{\timematrixname}{temporal simplicial matrix}

\title{Counting simplicial pairs in hypergraphs}

\author{
Jordan Barrett\thanks{Department of Mathematics, Toronto Metropolitan University, Toronto, Canada; e-mail: \texttt{jordan.barrett@torontomu.ca}}
\And
Pawe\l{} Pra\l{}at\thanks{Department of Mathematics, Toronto Metropolitan University, Toronto, Canada; e-mail: \texttt{pralat@torontomu.ca}}
\And
Aaron Smith\thanks{Department of Mathematics and Statistics, University of Ottawa, Ottawa, Canada; e-mail: \texttt{asmi28@uOttawa.ca}}
\And
Fran\c{c}ois Th\'{e}berge\thanks{Tutte Institute for Mathematics and Computing, Ottawa, Canada; email: \texttt{theberge@ieee.org}}
}

\begin{document}

\maketitle

\begin{abstract}
We present two ways to measure the simplicial nature of a hypergraph: the \measurename{} and the \matrixname{}. We show that the simplicial ratio captures the frequency, as well as the rarity, of simplicial interactions in a hypergraph while the simplicial matrix provides more fine-grained details. We then compute the \measurename{}, as well as the \matrixname{}, for 10 real-world hypergraphs and, from the data collected, hypothesize that simplicial interactions are more and more \emph{deliberate} as edge size increases. We then present a new Chung-Lu model that includes a parameter controlling (in expectation) the frequency of simplicial interactions. We use this new model, as well as the real-world hypergraphs, to show that multiple stochastic processes exhibit different behaviour when performed on simplicial hypergraphs vs.\ non-simplicial hypergraphs.
\end{abstract}

\section{Introduction}\label{sec:intro}

Many datasets that are typically represented as graphs would be more accurately represented as hypergraphs. For example, in the graph representation of a collaboration dataset, authors are represented as vertices and an edge exists between two vertices if the corresponding authors wrote a paper together~\cite{odda2006ramsey}. Using this representation, it is impossible to distinguish between a three-author paper and three separate two-author papers. In contrast, when we represent a collaboration dataset as a hypergraph we can clearly distinguish between a three-author paper (a single hyperedge) and three separate two-author papers (three distinct hyperedges). Hypergraph representations have proven to be useful for studying collaboration datasets~\cite{juul2022hypergraph}, protein complexes and metabolic reactions~\cite{feng2021hypergraph}, and many other datasets that are traditionally represented as graphs~\cite{lee2024survey}. Moreover, after many years of intense research using graph theory in modelling and mining complex networks~\cite{easley2010networks,jackson2010social,kaminski2021mining,newman2018networks}, hypergraph theory has started to gain considerable traction~\cite{battiston2020networks,benson2018simplicial,benson2021higher,benson2016higher,kamiński2024modularitybasedcommunitydetection,kaminski2019clustering,kaminski2023hypergraph}. It is becoming clear to both researchers and practitioners that higher-order representations are needed to study datasets involving higher-order interactions~\cite{benson2021higher,lambiotte2018understanding,tian2024higher,lee2024survey}.

Similar to hypergraph representations, simplicial complexes provide another way to represent datasets with higher-order interactions and, in some cases, it is not clear what the better model is for a given dataset~\cite{Kim23contagion,Torres21representation,Zhang23interactions}. The notion of \textit{simpliciality} was first introduced by Landry, Young and Eikmeier in~\cite{landry2024simpliciality} as a way of describing how closely a hypergraph resembles its simplicial closure. In their work, they discover that many hypergraphs built from real-world data, although not actually simplicial complexes, resemble their simplicial closures more closely than random hypergraphs. In a similar but distinct study, LaRock and Lambiotte in~\cite{LaRock23encapsulation} find that real-world hypergraphs often contain more instances of hyperedges contained in other hyperedges than in random hypergraphs. The results found in these two papers suggest that real-world hypergraphs are organized in a way where many of the small hyperedges live inside larger hyperedges. In our work, we pursue this idea further and define a ratio and a matrix for hypergraphs, which we call the \textit{simplicial ratio} and \textit{simplicial matrix} respectively, based on the number of instances of hyperedges inside other hyperedges compared to that of a null model. 

The remainder of the paper is organized as follows. In Sections~\ref{subsec:notation} and~\ref{subsec: old measures} we discuss notation as well as the measures for simpliciality given in~\cite{landry2024simpliciality}. Next, we define the simplicial ratio in Section~\ref{subsec:measure}, the simplicial matrix in Section~\ref{subsec:matrix}, and temporal variants in Section~\ref{subsec:temporal}. Then, in Section~\ref{subsec:data} we compute the \measurename{} and \matrixname{} of the same 10 real-world hypergraphs that were studied in~\cite{landry2024simpliciality} and then analyse this data in Section~\ref{subsec:analysis}. In Section~\ref{sec:model} we present a new random graph model that allows for more or less instances of hyperedges inside other hyperedges depending on an input parameter $q \in [0,1]$. In Section~\ref{sec:experiments} we experiment with four stochastic processes, comparing the processes on real-world hypergraphs and on our proposed model for varying $q$. Finally, we conclude and suggest further research in Section~\ref{sec:conclusion}.

\subsection{Notation}\label{subsec:notation}

In this paper, we use the terms graph and edge in lieu of hypergraph and hyperedge. 

A graph $G$ is a pair $(V(G), E(G))$ where $V(G)$ is a set of vertices and $E(G)$ is a collection of edges, i.e., a collection of subsets of vertices. We insist that $\emptyset \notin E(G)$ for any graph $G$. In general, for a graph $G$ and edge $e \in E(G)$, it is acceptable that $|e| = 1$. In this paper, however, we forbid such edges and consider only edges of size at least 2. We write $[n] := \{1,\dots,n\}$ and typically label the vertices in $G$ as $[n]$. A subgraph of a graph $G$ is any graph $H = (V(H), E(H))$ with $V(H) \subseteq V(G)$ and $E(H) \subseteq E(G)$ (note that, as $H$ is itself a graph, any edge $e \in E(H)$ contains only vertices in $V(H)$). For $e \in E(G)$, write $|e|$ for the size of $e$ and, for each positive integer $k$, define
\[
E_k(G) := \{e \in E(G) , |e| = k\} \,.
\]
If $E_k(G) = E(G)$ for some $k > 0$, then we call $G$ a $k$-uniform graph. Note that, for any graph $G$, the graph $G_k := (V(G),E_k(G))$ is a $k$-uniform subgraph of $G$, and 
\[
G = \bigcup_{k > 0} G_k \,,
\]
and thus every graph is the edge-disjoint union of uniform subgraphs. 

A \textit{multigraph} $G$ is a graph that allows edges $e \in E(G)$ with more than one instance of the same vertex (multiset edges) and allows multiple edges $e_1,\dots,e_k \in E(G)$ that are identical (parallel edges); a graph $G$ is \textit{simple} if it contains no multiset edges or parallel edges. Note that all simple graphs are multigraphs. For a multigraph $G$ and a vertex $v$, writing $m_G(v, e)$ for the number of instances of $v$ in $e$, the \textit{degree of $v$ in $G$}, denoted $\mathrm{deg}_G(v)$, is defined as 
\[
\mathrm{deg}_G(v) := \sum_{e \in E(G)} m_G(v, e)\,.
\]
If $G$ is simple, we equivalently have
\[
\mathrm{deg}_G(v) = \Big| \big\{ e \in E(G) ~\big|~ v \in e \big\} \Big| \,.
\]
All graphs in this paper are simple except for the random graphs generated by Algorithm~\ref{alg:chung lu} and Algorithm~\ref{alg:simplicial}.

We use standard notation for probability, i.e., $\p{\cdot}$ for probability, $\E{\cdot}$ for expectation. We write $X \sim \cU$ to mean $X$ is sampled from distribution $\cU$ and write $X_1, \dots, X_k \stackrel{i.i.d.}{\sim} \cU$ to mean $X_1, \dots, X_k$ are sampled independently and identically from distribution $\cU$. For a set $S$, we write $X \in_u S$ to mean that $X$ is chosen uniformly at random from $S$.

\subsection{Measures for simpliciality} \label{subsec: old measures}

In~\cite{landry2024simpliciality}, Landry, Young and Eikmeier establish three distinct measures quantifying how close a graph is to a simplicial complex. The first measure they establish is the \textit{simplicial fraction}. Given a graph $G$, let $S \subseteq E(G)$ be the set of edges such that $e \in S$ if and only if $|e| \geq 3$ and, for all $f \subseteq e$ with $|f| \geq 2$, $f \in E(G)$. Then the \textit{simplicial fraction} of $G$, written $\sigma_\mathrm{SF}(G)$, is defined as
\[
\sigma_\mathrm{SF}(G) := \frac{|S|}{\left| \bigcup_{k \geq 3} E_k(G)\right|} \,.
\]
In words, $\sigma_\mathrm{SF}(G)$ is the proportion of edges of size at least 3 in $E(G)$ that satisfy downward closure. 

The second and third measures that Landry, Young and Eikmeier establish are the \textit{edit simpliciality} and the \textit{face edit simpliciality}, respectively. For a graph $G$, define the \textit{$k$-closure}, written $\overline{G}_k$, as the graph $(V(\overline{G}_k), E(\overline{G}_k))$ where 
\begin{align*}
V(\overline{G}_k) &= V(G),\\
E(\overline{G}_k) &= \Big\{e \subseteq V(G) \, \Big| \, |e| \geq k \text{ and } e \subseteq f \text{ for some } f \in E(G) \Big\}. 
\end{align*} 
Then the \textit{edit simpliciality} of $G$, written $\sigma_\mathrm{ES}(G)$, is defined as 
\[
\sigma_\mathrm{ES}(G) := \frac{|E(G)|}{|E(\overline{G}_2)|} \,.
\]
Thus, $1 - \sigma_\mathrm{ES}(G)$ is the (normalized) number of additional edges needed to turn $G$ into its $2$-closure. Similarly, the \textit{face edit simpliciality} of $G$, written $\sigma_\mathrm{FES}(G)$, is the average edit simpliciality across all induced subgraphs defined by maximal edges (edges not contained in other edges) in $\bigcup_{k \geq 3} E_k(G)$.

Using the three measures defined above, Landry, Young and Eikmeier show that real-world graphs are significantly more simplicial than graphs sampled from random models. However, they also note some unique short-comings of each measure. In the following two examples, we show some additional short-comings that are shared among all three measures. The first example shows that none of the measures properly capture the \textit{types} of simplicial relationships in a graph. 

\begin{example}\label{ex:short-coming}
Fix $n,k$ with $5 \leq k$ and $3k \leq n$. Let $G_1$ be a graph on the vertex set $[n]$ and with three edges $\{1,\dots,k\}, \{k+1,\dots,2k\}, \{2k+1,\dots,3k\}$ of size $k$ and three edges $\{1,2,3\}, \{k+1,k+2,k+3\}, \{2k+1,2k+2,2k+3\}$ of size 3. Let $G_2$ be a graph on the same vertex set and with the same three edges $\{1,\dots,k\}, \{k+1,\dots,2k\}, \{2k+1,\dots,3k\}$ of size $k$, but now with three edges $\{1,\dots,k-1\}, \{k+1,\dots,2k-1\}, \{2k+1,\dots,3k-1\}$ of size $k-1$. See Figure~\ref{fig:short-coming} for an illustration of $G_1$ and $G_2$ with $n=18$ and $k=6$.

\begin{figure}[ht]
\[
\includegraphics[scale=0.6, trim = {1cm 11cm 1cm 4cm}]{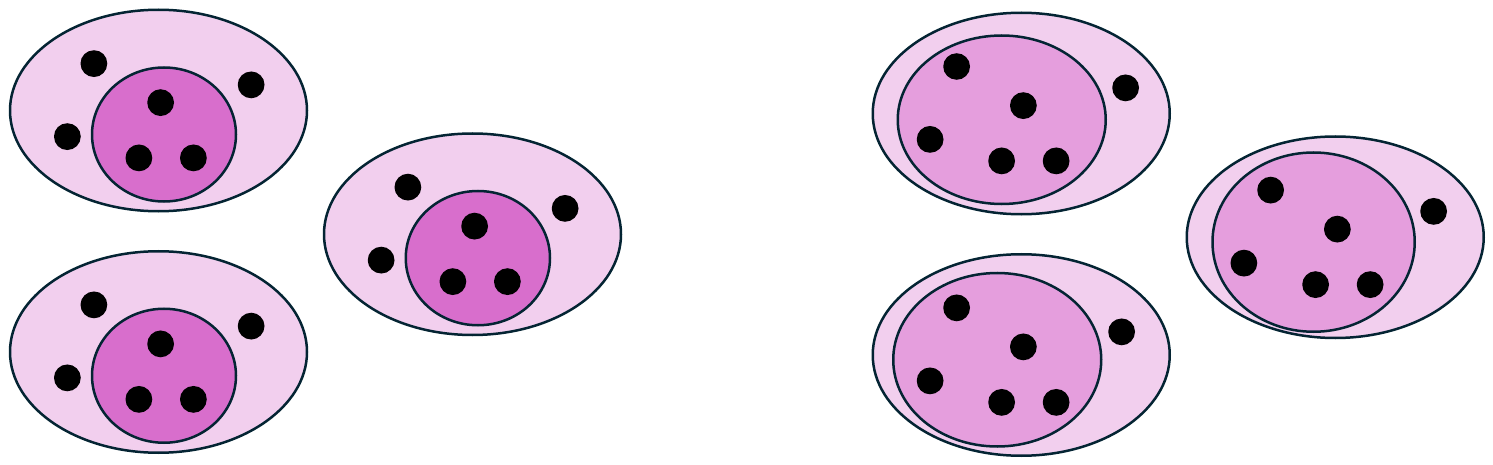}
\]
\caption{(left) a graph $G_1$ with 18 vertices, 3 edges of size 6, and 3 edges of size 3, and (right) a graph $G_2$ with 18 vertices, 3 edges of size 6, and 3 edges of size 5. We have $\sigma_{\mathrm{SF}}(G_1) = \sigma_{\mathrm{SF}}(G_2) = 0$, $\sigma_{\mathrm{ES}}(G_1) = \sigma_{\mathrm{ES}}(G_2) = 2/57$, and $\sigma_{\mathrm{FES}}(G_1) = \sigma_{\mathrm{FES}}(G_2) = 2/57$.}
\label{fig:short-coming}
\end{figure}

With $G_1$ and $G_2$ as defined above, we have 
\begin{align*}
\sigma_\mathrm{SF}(G_1) = \sigma_\mathrm{SF}(G_2) &= 0 \,,\\
\sigma_\mathrm{ES}(G_1) = \sigma_\mathrm{ES}(G_2) &= \frac{2 \cdot 3}{ (2^k-k-1) \cdot 3} = \frac{2}{2^k-k-1} \,, \text{ and}\\
\sigma_\mathrm{FES}(G_1) = \sigma_\mathrm{FES}(G_2) &= \frac{2}{2^k-k-1} \,,
\end{align*}
the value $2^k-k-1$ coming from the fact that there are $2^k$ subsets, $k$ of which are subsets of size 1, and 1 of which is the empty set. Thus, by all three measures, $G_1$ and $G_2$ are equally simplicial. However, qualitatively the simplicial relationships in $G_1$ are different than in $G_2$. Consider, for example, edges $e_3, e_5, e_6$ in an \ER random graph on $n$ vertices with $|e_3| = 3, |e_5| = 5$ and $|e_6| = 6$. Then, the probability of $e_3 \subset e_6$ (as in $G_1$) is of order $n^{-3}$, whereas the probability of $e_5 \subset e_6$ (as in $G_2$) is of order $n^{-5}$. 
\end{example}

The second example shows that, while the three measures are good indicators of how close a graph is to its $2$-closure, none of the measures are good indicators of how common it is to see edges inside of other edges in the graph. 

\begin{example}\label{ex:short-coming 2}
Let $G_1$ and $G_2$ be as shown in Figure~\ref{fig:short-coming_2}. There is a clear, strong simplicial structure in $G_1$, and there is clearly no simplicial structure in $G_2$. However, in both graphs, the simplicial fraction is 0 (none of the edges satisfy downward closure). Moreover, the edit simpliciality of $G_1$ is $4/57 \approx 0.07$ and of $G_2$ is $3/41 \approx 0.07$. Likewise, the face edit simpliciality of $G_1$ is $4/57 \approx 0.07$ and of $G_2$ is 
\[
\frac{1}{3}\left(\frac{1}{26} + \frac{1}{11} + \frac{1}{4} \right) \approx 0.13 \,. 
\] 
\end{example}

\begin{figure}[ht]
\[
\includegraphics[scale=0.5, trim = {1cm 12cm 1cm 4cm}]{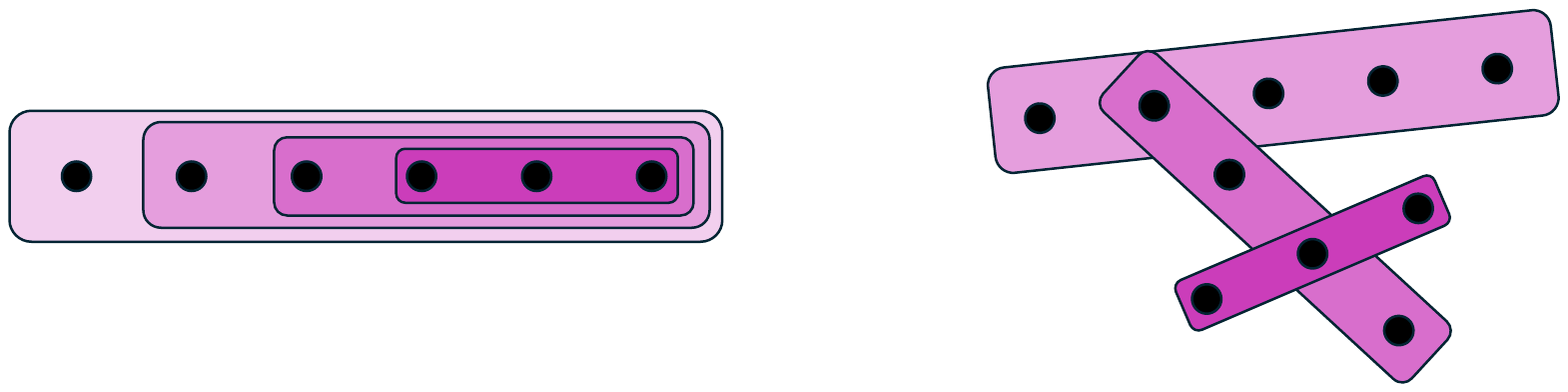}
\]
\caption{(left) a graph $G_1$ with 6 vertices and 4 edges, and (right) a graph $G_2$ with 10 vertices and 3 edges. We have $\sigma_{\mathrm{SF}}(G_1) = 0$, $\sigma_{\mathrm{ES}}(G_1) \approx 0.07$, $\sigma_{\mathrm{FES}}(G_1) \approx 0.07$, and $\sigma_{\mathrm{SF}}(G_2) = 0$, $\sigma_{\mathrm{ES}}(G_2) \approx 0.07$, $\sigma_{\mathrm{FES}}(G_2) \approx 0.13$.}
\label{fig:short-coming_2}
\end{figure}

Thus, $G_1$ and $G_2$ are equally simplicial according to the simplicial fraction and the edit simpliciality and, more strikingly, $G_1$ is \emph{less} simplicial than $G_2$ according to the face edit simpliciality. 

As mentioned previously, Examples~\ref{ex:short-coming} and~\ref{ex:short-coming 2} are not issues when we treat the simplicial fraction, edit simpliciality, and face edit simpliciality as measures of how close a graph is to its $2$-closure (as was their intended purpose). Instead, these examples suggest that if we want to understand the extent to which edges sit inside other edges in real-world networks then we need a new type of scoring system. 

\section{A new approach to simpliciality}\label{sec:definitions}

We aim to quantify a graph based on the frequency and rarity of edges inside other edges when compared to a null model. The metrics we present focus on the regime where data is ``slightly'' more simplicial than random (and so nearly-complete large simplices are extremely rare), while previous metrics focus on the regime where data is ``almost completely'' simplicial. The motivation behind these metrics is that the former regime is often more appropriate in real networks.

\subsubsection*{The hypergraph Chung-Lu model}

In the material to come, we frequently reference the hypergraph Chung-Lu model. The original model was defined for graphs~\cite{chung2006complex} and has been extensively studied since then. More recently, the model was generalized to other structures, including geometric graphs~\cite{kaminski2020unsupervised,kaminski2022multi} (both undirected and directed variants) as well as hypergraphs~\cite{kaminski2019clustering}. We give an algorithm for building the hypergraph model, conditioned on the number of edges, and point the reader to~\cite{kaminski2019clustering} for a full description of the model.

Let $(d_1,\dots,d_n)$ be a degree sequence on vertex set $[n]$ and let $(m_{k_{\mathrm{min}}},\dots,m_{k_{\mathrm{max}}})$ be a sequence of edge sizes where $m_k$ represents the number of edges of size $k$. Then, writing $p(\cdot)$ for the probability distribution with $p(v) = d_v / \sum_{u \in [n]} d_u$ for all $v \in [n]$, we first give the algorithm that generates a Chung-Lu edge of a given size.

\begin{algorithm}[ht]
\caption{Chung-Lu edge. }\label{alg:chung lu edge}
\begin{algorithmic}[1]
\Require $(d_1,\dots,d_n)$, $k$
\State Sample $e[1],\ldots,e[k] \stackrel{i.i.d.}{\sim} p(\cdot)$.
\State Return $\{e[1],\dots,e[k]\}$
\end{algorithmic}
\end{algorithm}

We now give the algorithm that generates a Chung-Lu graph. 

\begin{algorithm}[ht]
\caption{Chung-Lu Model. }\label{alg:chung lu}
\begin{algorithmic}[1]
\Require $(d_1,\dots,d_n)$, $(m_{k_{\mathrm{min}}},\dots,m_{k_{\mathrm{max}}})$.
\State Initialize edge list $E = \{\}$.
\For{ $k \in \{k_{\mathrm{min}},\dots ,k_{\mathrm{max}}\}$}
\For{ $i \in [m_k]$}
\State sample $e \sim \mathbf{Algorithm~\ref{alg:chung lu edge}}\Big((d_1,\dots,d_n), k\Big)$.
\State Set $E = E \cup \{e\}$.
\EndFor 
\EndFor 
\State Return $G = ([n], E)$.
\end{algorithmic}
\end{algorithm}

For a graph $G$ with degree sequence $\mathbf{d} = (d_1,\dots,d_n)$ and edge size sequence $\mathbf{m} = (m_{k_{\mathrm{min}}},\dots,m_{k_{\mathrm{max}}})$, we write $\hat{G} \sim \mathrm{CL}(G)$ to mean $\hat{G} \sim \mathrm{CL}(\mathbf{d}, \mathbf{m})$, where $\mathrm{CL}(\mathbf{d}, \mathbf{m})$ is the random graph returned by {Algorithm~\ref{alg:chung lu}}. A key feature of the Chung-Lu model is that the degree sequence is preserved in expectation. 
\begin{lemma}
Let $\hat{G} \sim \mathrm{CL}(G)$ for some graph $G$. Then 
\[
\E{\mathrm{deg}_{\hat{G}}(v)} = \mathrm{deg}_G(v)
\]
for all $v \in [n]$. 
\end{lemma}

\begin{proof}
Let $d_v := \mathrm{deg}_G(v)$ for all $v \in [n]$. First, note that every vertex in every edge of $\hat{G}$ is sampled independently with probability $p$, where $p(v) = \frac{d_v}{\sum_{u \in [n]} d_u}$. Thus, the expected total occurrence of $v$ in $E(\hat{G})$ is 
\begin{align*}
p(v) \sum_{e \in E(G)} |e| 
&=
\left( \frac{d_v}{\sum_{u \in [n]} d_u} \right) \sum_{e \in E(G)} |e| 
=
\left( \frac{d_v}{\sum_{u \in [n]} d_u} \right) \sum_{u \in [n]} d_u 
=
d_v \,,
\end{align*}
the second equality coming from the hypergraph counterpart of the hand-shaking lemma. Given that the total occurrence of $v$ in $E(\hat{G})$ is precisely $\mathrm{deg}_{\hat{G}}(v)$, the lemma follows. 
\end{proof}

\subsection{The \measurename}\label{subsec:measure}

We are now ready to define the graph quantity at the heart of this paper. In essence, this quantity tells us how surprising it is to see the number of simplicial pairs in a given graph.

For a graph $G$, a \textit{simplicial pair in $G$} is a pair of distinct edges $e_1,e_2 \in E(G)$ with $e_1 \subset e_2$. Let $\pairs{G}$ be the number of simplicial pairs in $G$. 

Let $G$ be a graph and let $\hat{G} \sim \mathrm{CL}(G)$ conditioned on $\hat{G}$ having no multiset edges. Then the \textit{\measurename{}}, denoted by $\measure{G}$, is defined as
\[
\measure{G} := \frac{\pairs{G}}{\E{\pairs{\hat{G}}}} \,,
\]
if $\E{\pairs{\hat{G}}} > 0$, and $\measure{G} := 1$ otherwise. In words, $\measure{G}$ is the ratio of the number of simplicial pairs to the expected number of simplicial pairs.

\begin{remark}
If $\E{\pairs{\hat{G}}} = 0$ then it is necessarily the case that $\pairs{G} = 0$, since it is always true that $\p{\hat{G} = G} > 0$. Moreover, if $\pairs{G} = 0$ and $\E{\pairs{\hat{G}}} = 0$ then the number of simplicial pairs is as expected and so we define $\measure{G} = 1$. 
\end{remark}

\begin{remark}
We have mentioned already that the sizes of the edges in a simplicial pair are important. For this reason, we condition on $\hat{G} \sim \mathrm{CL}(G)$ having no multiset edges. 
\end{remark}

\begin{remark}
Our choice of the Chung-Lu model is not necessary for defining the \measurename{}. One could equivalently define the \measurename{} by taking expectations with respect to any model: the configuration model, Erd\H{o}s-R\'{e}nyi model, Stochastic Block Model, ABCD model, etc. We choose to use the Chung-Lu model as, in our opinion, it achieves the best balance of (a) retaining important features of a graph and (b) allowing for fast approximations of $\E{\pairs{\hat{G}}}$. 
\end{remark}

\begin{remark}
As mentioned in the previous remark, we \textit{approximate} $\E{\pairs{\hat{G}}}$ rather than compute this expectation exactly. For a graph $G$, computing $\E{\pairs{\hat{G}}}$ is quite difficult as we discuss in the open problems presented in Section~\ref{subsec:further research}. We approximate using a Monte Carlo estimator which is detailed in Appendix~\ref{app: algs}.
\end{remark}

\subsubsection*{Examples}

Let us revisit Examples~\ref{ex:short-coming} and~\ref{ex:short-coming 2}. 

Starting with Example~\ref{ex:short-coming}, the number of simplicial pairs in both graphs is 3. However, in $G_1$ the expected number of simplicial pairs is $\approx 0.3$, and in $G_2$ this expectation is $\approx 0.008$. Thus, $\measure{G_1} \approx 10$, whereas $\measure{G_2} \approx 380$, suggesting that the number of simplicial relationships in $G_2$ is far more surprising than in $G_1$. This result confirms that the simplicial ratio weighs different types of simplicial pairs differently.

Continuing with Example~\ref{ex:short-coming 2}, we have that $\pairs{G_1} = 6$ and $\E{\pairs{\hat{G}}} \approx 4.3$, meaning $\measure{G_1} \approx 1.4$, whereas $\pairs{G_2} = 0$ and $\E{\pairs{\hat{G_2}}} \approx 0.2 > 0$, meaning $\measure{G_2} = 0$. Thus, the simplicial ratio can clearly distinguish $G_1$ and $G_2$. 

By computing the \measurename{} of the graphs in Examples~\ref{ex:short-coming} and~\ref{ex:short-coming 2}, we see a clear distinction between the three measures given in \cite{landry2024simpliciality} and the \measurename{} that we present here: the simplicial fraction, edit simpliciality, and face edit simpliciality are all ways of measuring how close a graph is to its induced simplicial complex, whereas the \measurename{} is a way to measure how \textit{surprisingly simplicial} a graph is.

\subsection{The \matrixname}\label{subsec:matrix}

For a graph $G$, write $\pairs{G,i,j}$ for the number of simplicial pairs $(e_1,e_2)$ in $G$ with $|e_1|=i$ and $|e_2|=j$ with $i<j$. Then, letting $\hat{G} \sim \mathrm{CL}(G)$ conditioned on having no multiset edges, the \matrixname{} of $G$, denoted by $\measurematrix{G}$, is the partial matrix with cell $(i,j)$ equalling
\[
\measurematrix{G,i,j} := \frac{\pairs{G,i,j}}{\E{\pairs{\hat{G}, i, j}}}
\] 
whenever $i<j$ and $G$ contains edges of size $i$ and of size $j$ (and substituting 0 if there are no simplicial pairs of this type), and with cell $(i,j)$ being empty otherwise. 

\begin{remark}
We once again approximate $\E{\pairs{\hat{G},i,j}}$ via the sampling technique found in Appendix~\ref{app: algs}. 
\end{remark}

Intuitively, the \matrixname{} ``unpacks'' the \measurename{} and shows how powerful the simplicial interactions between edges of all different sizes are. More formally, the \matrixname{} and \measurename{} of $G$ satisfy the following weighted sum.
\[
\measure{G} = \sum_{i<j} w_{i,j} \cdot \measurematrix{G,i,j}
\]
where 
\[
w_{i,j} := \frac{\E{\pairs{\hat{G},i,j}}}{\E{\pairs{\hat{G}}}} \,, \hspace{1cm} \sum_{i<j} w_{i,j} = 1 \,.
\]
We will see in Section~\ref{sec:empirical} that the \matrixname{} reveals information about real-world graphs that the \measurename{} alone does not. In particular, a hypothesis we make in this paper, as suggest by these matrices, is that \textit{the composition of an edge in a real-world network becomes more dependent on simpliciality as the edge size increases.}

\subsubsection*{Examples}

We again revisit Examples~\ref{ex:short-coming} and~\ref{ex:short-coming 2}. In Example~\ref{ex:short-coming}, $\measurematrix{G_1}$ contains one non-empty cell, $(3,6)$, with value $\approx 10$, and $\measurematrix{G_2}$ contains one non-empty cell, $(5,6)$, with value $\approx 380$. 

Example~\ref{ex:short-coming 2} is more interesting as $G_1$ contains simplicial pairs of various types. For  $G_1$, we have
\[
\measurematrix{G_1} \approx
\left[
\begin{array}{cccccc}
\emptyset & \emptyset & \emptyset & \emptyset & \emptyset & \emptyset \\
\emptyset & \emptyset & \emptyset & \emptyset & \emptyset & \emptyset \\
\emptyset & \emptyset & \emptyset & \mathbf{3.8} & \mathbf{1.7} & \mathbf{1} \\
\emptyset & \emptyset & \emptyset & \emptyset & \mathbf{2.4} & \mathbf{1} \\
\emptyset & \emptyset & \emptyset & \emptyset & \emptyset & \mathbf{1} \\
\emptyset & \emptyset & \emptyset & \emptyset & \emptyset & \emptyset \\
\end{array}
\right] \,,
\]
and for $G_2$ we have 
\[
\measurematrix{G_2} \approx
\left[
\begin{array}{ccccc}
\emptyset & \emptyset & \emptyset & \emptyset & \emptyset \\
\emptyset & \emptyset & \emptyset & \emptyset & \emptyset \\
\emptyset & \emptyset & \emptyset & \mathbf{0} & \mathbf{0} \\
\emptyset & \emptyset & \emptyset & \emptyset & \mathbf{0} \\
\emptyset & \emptyset & \emptyset & \emptyset & \emptyset \\
\end{array}
\right] \,.
\]
The \matrixname{} for $G_1$ unpacks the information about its simplicial interactions. Indeed, the \measurename{} simply tells us that the number of simplicial pairs is 1.4 times more than expected. On the other hand, the \matrixname{} tells us that all 3 simplicial pairs involving the edge of size 6 are to be expected, whereas the other three simplicial pairs are at least somewhat surprising. We can also see that the existence of the $(3, 4)$ pair in $G_1$ is more surprising than the existence of the $(3, 5)$ pair, which is in turn more surprising than the existence of the $(3, 6)$ pair. In general, given a graph $G$ and distinct edge sizes $i<j<k$, if $G$ has the property that $|E_j(G)| \leq |E_k(G)|$ then it follows from the sampling process in {Algorithm~\ref{alg:chung lu edge}} that $\E{\pairs{G, i, j}} \leq \E{\pairs{G, i, k}}$. In the case of Example~\ref{ex:short-coming 2}, we have that $|E_4(G_1)| = |E_5(G_1)| = |E_6(G_1)| = 1$ and $\E{\pairs{G_1, 3, 4}} \approx 0.26$, $\E{\pairs{G_1, 3, 5}} \approx 0.59$, and $\E{\pairs{G_1, 3, 6}} = 1$. 

\subsection{Including a temporal element}\label{subsec:temporal}

Many networks (both real and synthetic) are not merely static graphs, but rather evolving process with edges forming over time. In these evolving processes, there are two distinct formations of a simplicial pair: either a small edge could form first, followed by a larger (superset) edge, or a large edge could form first, followed by a smaller (subset) edge. In the context of a collaboration graph, a ``bottom-up'' formation is a group of collaborators who invite more people for a future collaboration, whereas a ``top-down'' formation is a group who exclude some people for a future collaboration. At least in this context, there is a substantial difference between bottom-up simplicial pairs and top-down simplicial pairs, and we would ultimately like to know how different networks bias towards or against the two types of simplicial formations. For this reason, we include a version of the \measurename{} and of the \matrixname{} that accounts for time-stamped edges. In the definitions to come, we assume that no two edges are born at the exact same time. 

Let $G$ be an evolving graph with time-stamped edges $E(G) = (e_1,\dots,e_m)$ such that $e_i$ was generated before $e_{i+1}$ for all $1\leq i < m$. Next, let $\uppairs{G}$ be the number of simplicial pairs $(e_i,e_j)$ in $G$ with $i < j$ and $|e_i| < |e_j|$, and let $\downpairs{G}$ be the number of simplicial pairs $(e_i, e_j)$ with $i > j$ and $|e_i| < |e_j|$. Finally, let $\hat{G} \sim \mathrm{CL}(G)$ and assign a uniformly random ordering to the edges of $\hat{G}$. Then the \upname{} and \downname{} of $G$, denoted $\upmeasure{G}$ and $\downmeasure{G}$ respectively, are defined as 
\[
\upmeasure{G} := \frac{\uppairs{G}}{\E{\uppairs{\hat{G}}}}
\hspace{1cm} \text{and} \hspace{1cm}
\downmeasure{G} := \frac{\downpairs{G}}{\E{\downpairs{\hat{G}}}} \,.
\]

\begin{remark}
By symmetry, we have that $\E{\uppairs{\hat{G}}} = \E{\downpairs{\hat{G}}} =  \frac{1}{2} \cdot \E{\pairs{\hat{G}}}.$ Thus, we can equivalently define the \upname{} and \downname{} respectively as 
\[
\frac{2 \cdot \uppairs{G}}{\E{\pairs{\hat{G}}}} \, \hspace{1cm} \text{and} \hspace{1cm} \frac{2 \cdot \downpairs{G}}{\E{\pairs{\hat{G}}}} \,.
\]
\end{remark}

For the temporal version of the \matrixname{} we distinguish between bottom-up and top-down simplicial pairs by their location in the matrix. For a temporal graph $G$ with edge ordering $E(G) = (e_1,\dots,e_m)$ and for $k<\ell$, write $\uppairs{G,k,\ell}$ for the number of simplicial pairs $(e_i,e_j)$ such that $i<j$, $|e_i|=k$, and $|e_j|=\ell$. Likewise, write $\downpairs{G,k,\ell}$ for the number of simplicial pairs $(e_i,e_j)$ such that $i > j$, $|e_i|=k$ and $|e_j|=\ell$. Then the \timematrixname{}, denoted $\timematrix{G}$, is the partial matrix with cell $(k,\ell)$ equalling 
\[
\timematrix{G, k, \ell} := \frac{\uppairs{G,k,\ell}}{\E{\uppairs{\hat{G}, k, \ell}}} \,,
\]
cell $(\ell,k)$ equalling 
\[
\timematrix{G, \ell, k} := \frac{\downpairs{G,k,\ell}}{\E{\downpairs{\hat{G}, k, \ell}}} \,,
\]
for all valid $k<\ell$, and cells $(k,\ell)$ and $(\ell,k)$ being empty otherwise.

\section{Empirical results}\label{sec:empirical}

In this section, we compute the simplicial ratio and simplicial matrix, both with and without a temporal element where applicable, for the same 10 graphs that were analysed in~\cite{landry2024simpliciality}. We then comment on the data and build some hypotheses about the simplicial nature of real networks.

The 10 graphs are all taken from \cite{Landry_XGI_A_Python_2023} and full descriptions can be found there. We paraphrase and summarize the descriptions below.

\medskip
\noindent
\textbf{contact-primary-school:} a temporal graph where nodes are primary students and edges are instances of contact (physical proximity) between students. 

\medskip
\noindent
\textbf{contact-high-school:} the same as contact-primary-school except with high-school students.

\medskip
\noindent
\textbf{hospital-lyon:} the same as contact-primary-school and contact-high-school except with patients and health-care workers in a hospital. 

\medskip
\noindent
\textbf{email-enron:} a temporal graph where nodes are email-addresses and edges comprise the sender and receivers of emails. 

\medskip
\noindent
\textbf{email-eu:} the same as email-enron except built from a different organization.

\medskip
\noindent
\textbf{diseasome:} a static (non-temporal) graph where nodes are diseases and edges are collections of diseases with a common gene.

\medskip
\noindent
\textbf{disgenenet:} a static graph where nodes are genes and edges are collections of genes found in a disease. In other words, disgenenet is precisely the line-graph of diseasome.

\medskip
\noindent
\textbf{ndc-substances:} a static graph where nodes are substances and edges are collections of substances that make up various drugs.

\medskip
\noindent
\textbf{congress-bills:} a temporal graph where nodes are US Congresspersons and edges comprise the sponsor and co-sponsors of legislative bills put forth in both the House of Representatives and the Senate.

\medskip
\noindent
\textbf{tags-ask-ubuntu:} a temporal graph where nodes are tags and edges are collections of tags applied to questions on the website \url{askubuntu.com}.

\medskip
\noindent
For each graph, we restrict to edges of sizes 2 through 11, as is the case in \cite{landry2024simpliciality}
We throw away multi-edges, only keeping the first occurrence of each edge in the case of temporal graphs. We approximate $\E{\hat{G}}$ using our Chung-Lu sampling technique presented in Appendix~\ref{app: algs}. 

\subsection{The data} \label{subsec:data}
In Table \ref{tab:ten}, we show the simplicial ratios as well as useful information about each graph.
In Figure~\ref{fig:sm1} we show the simplicial matrices of these graphs and in Figure~\ref{fig:sm2} we show the temporal matrices of the 7 temporal graphs. For readability we show only the non-empty cells of the partial matrices and omit cells involving edges of size greater than 5. Figure~\ref{fig:example shorthand} shows the simplified presentation of the simplicial matrix of $G_1$ from Example~\ref{ex:short-coming 2}. 

\begin{table}[H]\label{tbl:all-sr}
\[
\begin{array}{|c|c|c|c|c|c|c|c|c|}
\hline
G & |V(G)| & |E(G)| & [|E_2|,|E_3|,|E_4|,|E_{\geq 5}|] & \measure{G} & \upmeasure{G} & \downmeasure{G}\\
\hline
\text{disgenenet} & 1982 & 760 & [157, 139, 93, 371] & 28.81 & \text{n.a.} & \text{n.a.} \\
\hline
\text{contact-h.s.} & 327 & 7818 & [5498, 2091, 222, 7] & 6.68 & 11.19 & 2.17 \\
\hline
\text{diseasome} & 516 & 314 & [153, 92, 26, 43] & 6.49 & \text{n.a.} & \text{n.a.} \\
\hline
\text{email-eu} & 967 & 23729 & [13\mathrm{k},5\mathrm{k},2\mathrm{k},4\mathrm{k}] & 5.19 & 5.77 & 3.72 \\
\hline
\text{email-enron} & 143 & 1442 & [809, 317, 138, 178] & 4.96 & 6.98 & 2.94 \\
\hline
\text{congress-bills} & 1715 & 58788 & [14\mathrm{k}, 10\mathrm{k}, 8\mathrm{k}, 27\mathrm{k}] & 4.46 & 5.23 & 3.69 \\
\hline
\text{ndc-substances} & 2740 & 4754 & [1130, 745, 535, 2344] & 4.22 & \text{n.a.} & \text{n.a.} \\
\hline
\text{contact-p.s.} & 242 & 12704 & [7748, 4600, 347, 9] & 2.74 & 4.82 & 0.66 \\
\hline
\text{hospital-lyon} & 75 & 1824 & [1107, 657, 58, 2] & 0.94 & 1.71 & 0.17 \\
\hline
\text{tags-ask-ubuntu} & 3021 & 145053 & [28\mathrm{k},52\mathrm{k},39\mathrm{k},25\mathrm{k}] & 0.69 & 1.09 & 0.29 \\
\hline
\end{array}
\]
\caption{The \measurename{} of 10 real networks and the corresponding \upname{} and \downname{} for the 7 temporal networks. The graphs are ordered according to $\measure{G}$, from largest to smallest.}
\label{tab:ten}
\end{table}

\medskip

\begin{figure}[ht]
\[
\includegraphics[scale=0.75, trim = {2cm 13cm 2cm 0.5cm}]{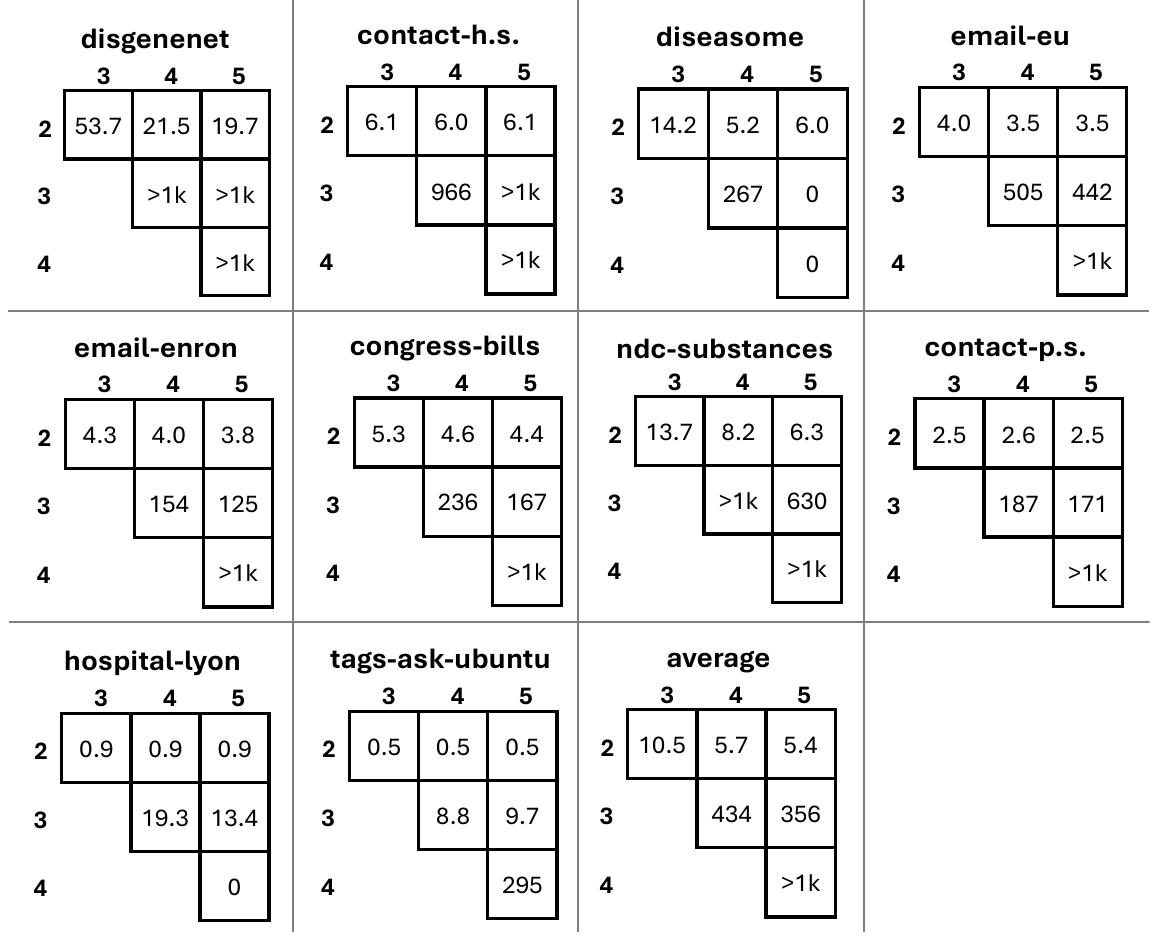}
\]
\caption{The \matrixname{} of 10 real networks, as well as the cell-wise average matrix. For each graph $G$, only non-empty cells of $\measurematrix{G}$ are shown, and cells involving edges of size greater than 5 are omitted. The value of a cell is replaced with ``$>1\mathrm{k}$'' whenever the value is above 1000.}
\label{fig:sm1}
\end{figure}

\clearpage 

\begin{figure}[ht]
\[
\includegraphics[scale=0.7, trim = {2cm 8cm 2cm 0.5cm}]{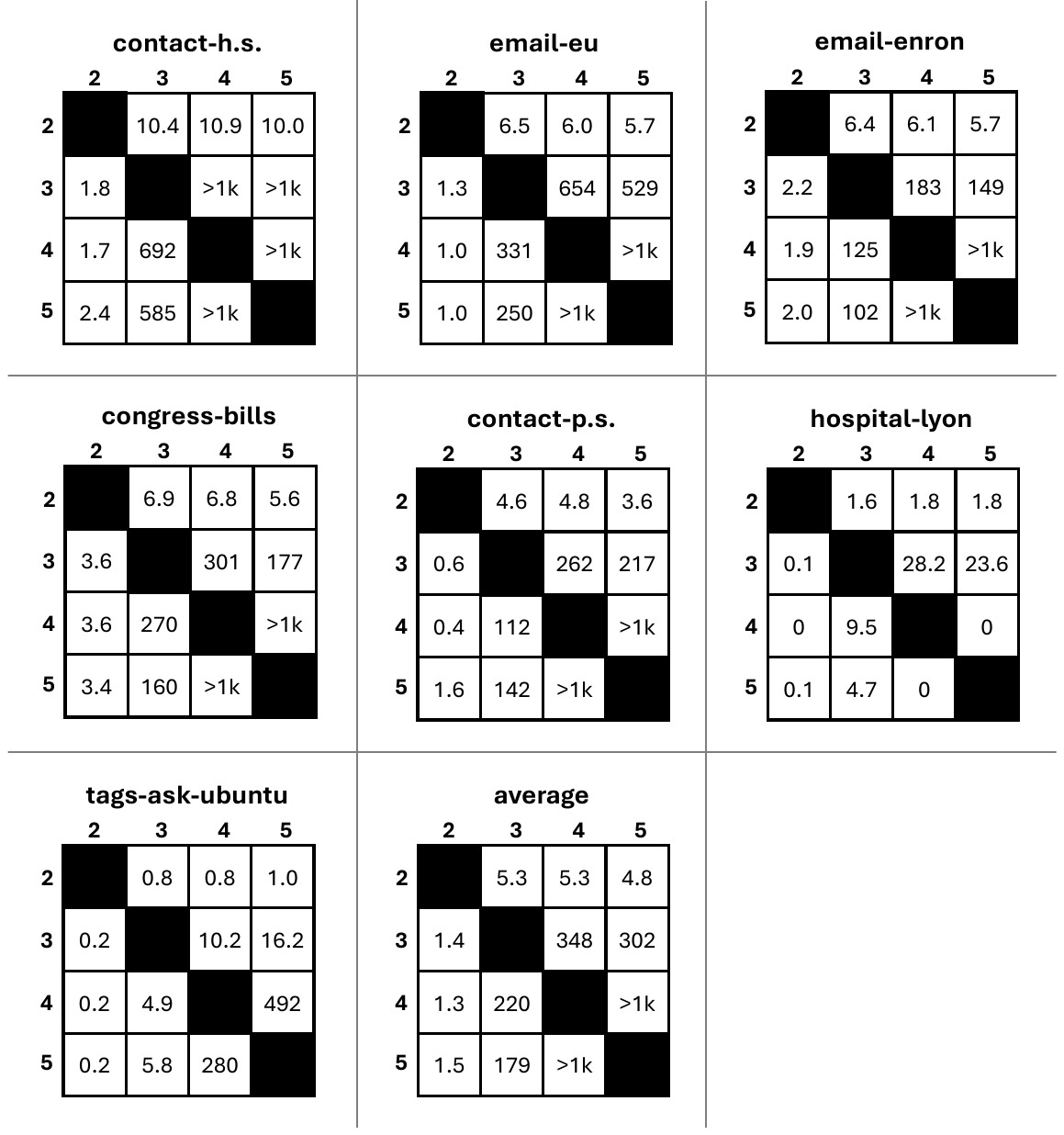}
\]
\caption{The \timematrixname{} of 7 real networks, as well as the cell-wise average matrix. For each graph $G$, only non-empty cells of $\timematrix{G}$ are shown, and cells involving edges of size greater than 5 are omitted. The value of a cell is replaced with ``$>1\mathrm{k}$'' whenever the value is above 1000.}
\label{fig:sm2}
\end{figure}

\begin{figure}[H]
\[
\includegraphics[scale=0.4]{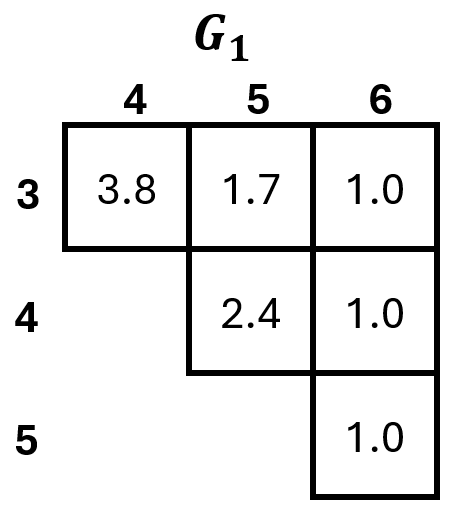}
\]
\caption{The simplicial matrix of $G_1$ from Example~\ref{ex:short-coming 2}, presented in a simplified manner.}
\label{fig:example shorthand}
\end{figure}

\clearpage 

\subsection{Analysis}\label{subsec:analysis}

\subsubsection*{Simplicial ratio}

Based on our results, we see that that biology networks are, on average, more surprisingly simplicial than contact-based networks and email networks. In contrast, it was shown in~\cite{landry2024simpliciality} that contact-based networks are the closest to their simplicial closures and biological networks are furthest from theirs. In fact, comparing the ranks of the 3 existing measures (sf, es, fes) and the ranks from our simplicial ratio (sr), we get the following Kendall correlation values.
\begin{center}
\begin{tabular}{lrrrr}
\toprule
 & sf & es & fes & sr \\
\midrule
sf & 1.000 & 0.706 & 0.989 & -0.270 \\
es & 0.706 & 1.000 & 0.722 & -0.256 \\
fes & 0.989 & 0.722 & 1.000 & -0.289 \\
sr & -0.270 & -0.256 & -0.289 & 1.000 \\
\bottomrule
\end{tabular}
\end{center} 
These values show that our ranking system is negatively correlated with the ranking systems in \cite{landry2024simpliciality}. A partial explanation for this correlation is that (a) the measures behave differently under different regimes of edge density and (b) the 10 datasets cover a wide range of edge density.

\subsubsection*{Bottom-up and top-down simplicial ratios}

In our testing, we find that every temporal graph contains more bottom-up simplicial pairs than top-down simplicial pairs. This suggests that, in general for many real networks, a small edge leading to a larger (superset) edge is more common than a large edge leading to a smaller (subset) edge.
However, this result is \textit{heavily} biased on our choice of keeping only the first instance of an edge. To see this bias, let $G$ be a temporal graph with edges $e_1, e_2 \in E(G)$ such that $e_1 \subset e_2$ and suppose that $e_1$ appears with multiplicity 5 and that $e_2$ appears with multiplicity 1. Then there are 6 possible birth orderings for $e_2$ and the 5 copies of $e_1$, and only 1 such ordering sees $e_2$ born before $e_1$. In most of the temporal networks analysed, the highest frequency of multi-edges are indeed 2-edges, and hence this bottom-up trend is at least partly explained by the above discussion. The topic of temporal simpliciality is one that we intend on exploring further in future works.

\subsubsection*{Simplicial matrix}

Arguably the most immediate take-away from these matrices is that simplicial interactions become less likely as edge size increases. Although this feature is interesting, there is at least a partial explanation for this phenomenon that we explore in the following example. 

\begin{example}\label{ex:size matters}
Let $n \in \N$, $\mathbf{d}$ be a uniform degree sequence, and let $\mathbf{m} = (m_2,\dots,m_5)$ be a sequence of edge sizes with $m_2=m_3=m_4=m_5=n$. Now let $G \sim \mathrm{CL}(n,\mathbf{p},\mathbf{m})$, and let $e_2,e_3,e_4,e_5 \in E(G)$ be chosen uniformly at random conditioned on $|e_i| = i$ for each $i \in \{2,3,4,5\}$. Then, writing $X_{i,j}$ for the indicator variable which is 1 if $e_i \subset e_j$, we have
\[ 
\begin{array}{l|l|l}
\E{X_{2,3}} \propto n^{-2} & \E{X_{2,4}} \propto n^{-2} & \E{X_{2,5}} \propto n^{-2}\\
\hline
& \E{X_{3,4}} \propto n^{-3} & \E{X_{3,5}} \propto n^{-3}\\
\hline
& & \E{X_{4,5}} \propto n^{-4} 
\end{array}
\]
which implies
\[ 
\begin{array}{l|l|l}
\E{\pairs{G,2,3}} \propto 1 & \E{\pairs{G,2,4}} \propto 1 & \E{\pairs{G,2,5}} \propto 1\\
\hline
& \E{\pairs{G,3,4}} \propto 1/n & \E{\pairs{G,3,5}} \propto 1/n\\
\hline
& & \E{\pairs{G,4,5}} \propto 1/n^2
\end{array}
\]
\end{example}
Now, let $H$ be a graph with degree sequence $\mathbf{d}$ and edge-size sequence $\mathbf{m}$, and suppose $H$ has one simplicial pair of each type. Then, based on the above calculations, we get that
\[ 
\begin{array}{l|l|l}
\measure{H,2,3} \propto 1 & \measure{H,2,4} \propto 1 & \measure{H,2,5} \propto 1\\
\hline
& \measure{H,3,4} \propto n & \measure{H,3,5} \propto n\\
\hline
& & \measure{H,4,5} \propto n^2
\end{array}
\]
Thus, the above matrix acts as a loose, point-wise lower-bound on the \matrixname{} for sparse graphs with at least one simplicial pair of each type. For many of the graphs analysed, this rough sketch of a \matrixname{} is a good approximation of the actual matrices. In summary, what the \matrixname{} is capturing, above all else, is that (a) real graphs contain simplicial pairs of all types, and (b) synthetic (sparse) models very rarely generate simplicial pairs other than pairs containing 2-edges. 

\subsubsection*{Temporal simplicial matrix}

In general, the bias towards bottom-up simplicial pairs (and top-down simplicial pairs in the ``tags-ask-ubuntu'' graph) is consistent with the cell-wise comparisons. This suggests that the bias is independent, or at least not heavily dependent, on edge size. 

\section{A new model that incorporates simpliciality}\label{sec:model}

In this section, we define a random graph model, called the \textit{simplicial Chung-Lu model}, that generalizes the Chung-Lu hypergraph model defined in \cite{kaminski2019clustering}. We begin with the algorithm that generates a simplicial edge.

Let $(d_1,\dots,d_n)$ be a degree sequence, $k$ be an edge size, $E$ be a set of existing edges, and $E_k \subseteq E$ be a set of existing edges that are of size $k$. Recalling that $p(\cdot)$ is the probability distribution governed by $(d_1,\dots,d_n)$, writing ${S \choose k}$ for the collection of $k$-subsets of $S$, and recalling that $x \in_u X$ means $x$ is sampled uniformly from $X$, the algorithm to generate a simplicial edge is as detailed in Algorithm \ref{alg:simplicial edge}. 

\begin{algorithm}[ht]
\caption{Simplicial edge. }\label{alg:simplicial edge}
\begin{algorithmic}[1]
\Require $(d_1,\dots,d_n)$, $k$, $E$.
\If{ $E \setminus E_k = \emptyset$}
\State Sample $e \sim \mathbf{Algorithm~1}\Big( (d_1,\dots,d_n), k \Big)$
\Else
\State Sample $e' \in_u E \setminus E_k$.
\If{ $|e'| < k$}
\State Sample $e'' \sim \mathbf{Algorithm~1}\Big( (d_1,\dots,d_n), k-|e'| \Big)$.
\State Set $e = e' \cup e''$
\Else
\State Sample $e \in_u {e' \choose k}$
\EndIf
\EndIf
\State Return $e$
\end{algorithmic}
\end{algorithm}

\noindent
In words, we first check if there is at least one edge in $E$ \textit{not} of size $k$ to pair $e$ with. If there is no such edge, we return a Chung-Lu edge. Otherwise, we choose an existing edge $e'$ uniformly at random from the set of edges \textit{not} of size $k$ and construct our edge $e$ from $e'$ in one of two ways: if $k<|e'|$ we set $e$ to be a uniform $k$-subset of $e'$, whereas if $k>|e'|$ we build $e$ by combining $e'$ with a Chung-Lu edge of size $k-|e'|$. 

In Algorithm \ref{alg:simplicial}, we deescribe how to generate a simplicial Chung-Lu graph. Let $(d_1,\dots,d_n)$ be a degree sequence, $(m_{k_{\mathrm{min}}},\dots,m_{k_{\max}})$ be a sequence of edge sizes, and $S = (s_1, \dots, s_\ell)$ be a random permutation of all available sizes for an edge, i.e., $S$ contains $m_k$ copies of $k$ for each edge size $k$ in some random order. Additionally, let $q \in [0,1]$ be a parameter governing the number of simplicial edges created during the process. 

\begin{algorithm}[ht]
\caption{Simplicial Chung Lu model.}\label{alg:simplicial}
\begin{algorithmic}[1]
\Require $(d_1,\dots,d_n)$, $(m_{k_{\mathrm{min}}},\dots,m_{k_{\mathrm{max}}})$, $q$.
\State Initialize edge list $E = \{\}$ and random edge-size list $S$.
\For{ $k \in S$}
\State Sample $X \sim \mathrm{Bernoulli}(q)$.    
\If{ $X = 1$}
\State Sample $e \sim \mathbf{Algorithm~3}\Big( (d_1,\dots,d_n), k, E \Big)$
\Else
\State Sample $e \sim \mathbf{Algorithm~1}\Big( (d_1,\dots,d_n), k \Big)$
\EndIf
\State Set $E = E \cup \{e\}$.
\EndFor 
\State Return $G = ([n], E)$.
\end{algorithmic}
\end{algorithm}

Note that, if $q=0$, the simplicial Chung-Lu model yields a Chung-Lu model, ensuring that this new model is indeed a generalized Chung-Lu model. Moreover, the following lemma shows that the main feature of the Chung-Lu model is still present in this new model. 

\begin{lemma}
Let $G$ be a random graph generated as a simplicial Chung-Lu model with input parameters $(d_1,\dots,d_n)$, $(m_{k_{\mathrm{min}}},\dots, m_{k_{\mathrm{max}}})$, and $q \in [0,1]$. Then, for all $v \in [n]$,
\[
\E{\mathrm{deg}_G(v)} = d_v \,.
\]
\end{lemma}

\begin{proof}
Let us generate a random edge-size list $S$ that will be used to create the simplicial Chung-Lu graph $G$. We will first prove (by induction on $i$) the following claim. 

\noindent \textbf{Claim}: Each vertex $v$ of the $i$'th edge $e_i$ formed during the construction process of $G$ satisfies
$$
\p{v = u} = p(u) \text{ for all } u \in [n].
$$
Note that edges of $G$ are not generated independently; the graph has rich dependence structure. The distribution of $e_i$ is affected by edges generated earlier. It is important to keep in mind that the claim applies to the edge formed at time $i$ but without exposing information about earlier edges.

Firstly, if $i=1$, then $e_1$ is necessarily generated via Algorithm~\ref{alg:chung lu edge} and the claim follows immediately. Now fix $i>1$ and consider the formation of $e_i$. On the one hand, if $e_i$ was generated via {Algorithm~\ref{alg:chung lu edge}} then the claim is once again immediate. Otherwise, $e_i$ was generated via {Algorithm~\ref{alg:simplicial edge}}, i.e., generated constructively from another edge $e_j$ with $j<i$. In this case, if $|e_i|<|e_j|$ then $e_i \in_u {e_j \choose |e_i|}$ and, regardless which subset of $e_j$ is selected to form $e_i$, the claim holds by induction. Otherwise, if $|e_i|>|e_j|$, then $e_i$ is the union of $e_j$ and another edge $e''$ generated via {Algorithm~\ref{alg:chung lu edge}}: the claim holds immediately for vertices in $e''$, and for vertices in $e_j$, the claim holds by induction. 

Thus, for any $e \in E(G)$, $v \in e$, and $u \in [n]$, we have that $\p{v = u} = p(u)$. Summing over all vertices in all edges, we get that 
\[
\E{\mathrm{deg}_G(u)} 
= 
\left(\sum_{e \in E(G)} \sum_{v \in e} \p{v = u}\right)
=
\left(p(u) \sum_{e \in E(G)} |e| \right)
=
\left(p(u) \sum_{v \in [n]} d_v\right)
= 
d_u \,,
\] 
the first equality following from linearity of expectation, and the third equality following from the generalized handshaking lemma. 
\end{proof}

The simplicial Chung Lu model does in fact generate more simplicial pairs as $q$ increases. Figure~\ref{fig:varying q} shows the expected number of simplicial pairs (approximated via 1000 samples) for graphs generated via {Algorithm~\ref{alg:simplicial}} with $q$ varying from $0$ to $1$ in $0.1$ increments.

\begin{figure}[ht]
    \[
    \includegraphics[scale=0.45]{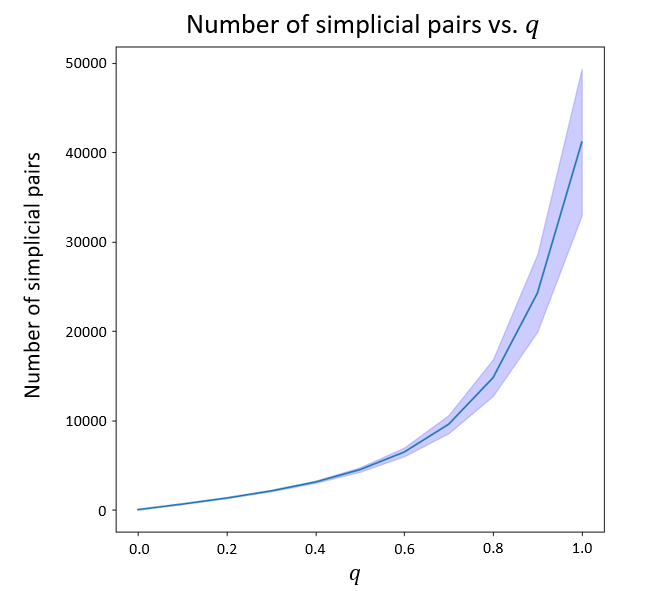}
    \hspace{0.2in}
    \]
    \caption{The average number of simplicial pairs (taken over 1000 samples) for simplicial Chung Lu graphs with varying $q$. For each $q \in [0, 0.1, \dots, 1]$, $G_q$ is a simplicial Chung Lu graph with $n = 1000$, $\mathbf{d}$ a uniform degree sequence, and $[|E_2|, |E_3|, |E_4|, |E_5|] = [5000, 1000, 100, 10]$. The shaded region represents the standard deviation over the 1000 samples.}
    \label{fig:varying q}
\end{figure}


\section{Experiments}\label{sec:experiments}

One reason to study simpliciality is that it likely has an impact on the evolution of stochastic processes on the associated graphs. We illustrate this potential impact via two toy processes with varying parameters. The first process is component growth which a standard way to measure the robustness of a network (see, for example, Chapter~8 in \cite{barabasi2016network}). The second type is information diffusion which simulates how quickly a substance (e.g., a disease, a rumour) spreads through a network. Intuitively, both of these processes should be affected by a graph containing a large number of simplicial pairs: in the case of component growth the smaller edge of a simplicial pair does not contribute to component size, and in the case of information diffusion a simplicial pair transfers information less efficiently than two non-overlapping edges.

\subsection{Descriptions of the experiments}\label{subsec:descriptions}

We perform four experiments (two experiments for each of the two types of stochastic processes) on the real networks and on the corresponding simplicial Chung-Lu graphs for varying $q \in \{0, 0.5, 1\}$. 

\bigskip
\noindent
\textbf{Giant component growth with random edge selection:} We choose a uniform random order for $E(G)$ and track the size of the largest component as edges are added to $G$ according to a random ordering. We plot the growth up to the point where $\min\{|E(G)|,|V(G)|\}$ edges have been added. We perform this experiment independently 10000 times on the real graphs, meaning we shuffle the edge ordering and track the growth 10000 times. For the simplicial Chung-Lu models we (a) sample the graph, (b) shuffle the edges, and (c) track the growth, performing steps (a), (b), and (c) independently 10000 times. 

\bigskip
\noindent
\textbf{Giant component growth with adversarial edge selection:} We order $E(G)$ in ascending order of betweenness (breaking ties randomly) and track the size of the largest component as edges are added to $G$ according to this adversarial ordering. Note that the betweenness of an edge $e$ in a hypergraph is equivalent to the betweenness of its corresponding vertex $v_e$ in the line graph (see~\cite{freeman1977set}, or any textbook on network science such as~\cite{kaminski2021mining}, for a definition of betweenness for graphs). For the real graphs, we run the experiment only once (the results will be the same every time), and for the Chung-Lu models we sample and track growth 20 times. We sample significantly less here than in the other three experiments due to the time complexity of calculating betweenness. 

\bigskip
\noindent
\textbf{Information diffusion from a single source:} We initialize a function $w_0:V(G) \to [0,1]$ with $w_0(v) = 0$ for all vertices, except for one randomly chosen vertex $v^*$ which has $w_0(v^*) = 1$. Then, in round $i+1$, we choose a random edge $e$ and, for each $v \in e$, set $w_{i+1}(v) = w(e)/|e|$, where $w(e) = \sum_{u \in e} w(u)$ (keeping $w_{i+1}(v) = w_{i}(v)$ for all $v \notin e$). We track the Wasserstein-1 distance (also known as the ``earth mover's distance'' \cite{710701}) between $w_i$ and the uniform distribution $w_\infty: V(G) \to 1/|V(G)|$. We run the experiment 10000 times, re-rolling the Chung-Lu model every time. 

\bigskip
\noindent
\textbf{Information diffusion from 10\% of the vertices:} This experiment is identical to the previous experiment, except that $w_0(v^*) = 1$ for 10\% of the vertices chosen at random, and that $w_\infty: V(G) \to 1/10$. 

\subsubsection*{Insisting on connected graphs}

These experiments, and in particular the two diffusion experiments, are highly dependent on connectivity. The real graphs are restricted to their largest component, and so we insist that the random graphs are also connected. To achieve this, we modify the simplicial Chung-Lu model and insist that incoming edges must connect disjoint components, until the point the graph is connected when we continue generating edges as normal. A full description of this algorithm is presented in Appendix~\ref{app: algs}. 

\subsection{The results}\label{subsec:experiment results}

Here, we will show the results for the two graphs: \textbf{hospital-lyon} and \textbf{disgenenet}. Recall that the \textbf{hospital-lyon} graph has a simplicial ratio of approximately 0.97, whereas the \textbf{disgenenet} graph has a ratio of approximately 15.99. The full collection of results can be found in Appendix~\ref{app: experiments} and the sampling technique can be found in Appendix~\ref{app: algs}.

\subsubsection*{Experiment 1: random growth}

In this first experiment shown in Figure \ref{fig:exp1}, we see the following. For \textbf{hospital-lyon} the real graph grows in a near identical way to the random model with $q=0$ and $q=0.5$, whereas the random model with $q=1$ grows much slower. In contrast, for \textbf{disgenenet} the real graph grows most similarly to the random model with $q=1$ whereas the random model with $q=0.5$ grows slightly faster, and for $q=0$ even faster still. Of course, these graphs have very different growth behaviour due to the difference in edge densities. Nevertheless, this result suggests that the high simplicial ratio of \textbf{disgenenet} plays a role in slowing down the growth of the graph, whereas the low simplicial ratio of \textbf{hospital-lyon} leads it to grow as quickly as a classical Chung-Lu model. 

\begin{figure}[ht]
\[
\includegraphics[scale=0.7,trim = {5cm 1cm 5cm 0cm}]{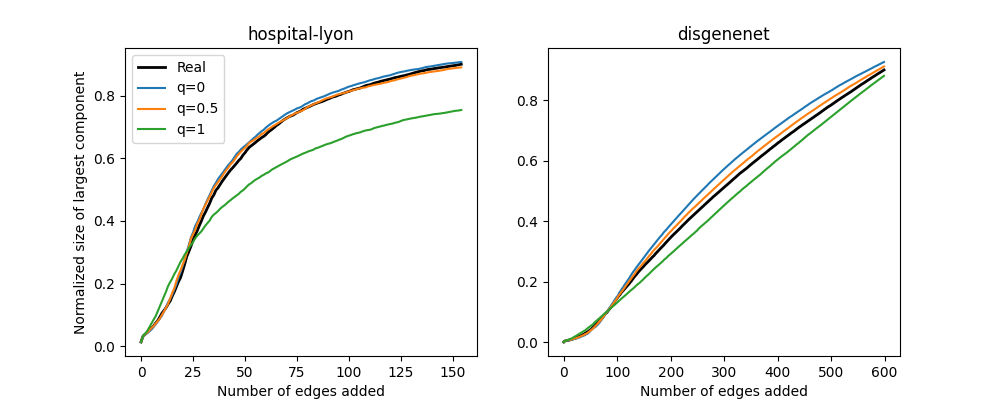}
\]
\caption{Giant component size (normalized by the number of vertices) vs.\ number of edges added in the random growth process on the \textbf{hospital-lyon} graph (left) and the \textbf{disgenenet} graph (right). The curve is the point-wise average across 10000 independent experiments: for the real graph the edges are resampled each time, and for the random models the entire graphs are resampled each time.}
\label{fig:exp1}
\end{figure}

\subsubsection*{Experiment 2: adversarial growth}

The results of this second experiment shown in Figure \ref{fig:exp2}, adversarial growth, are less clear due to the fact that we averaged over 20 samples instead of 10000. Nonetheless, there is still a clear distinction between the real growth vs.\ the synthetic growth for these two graphs. On the left, we see that the real graph grows faster than all the random models, whereas on the right the real graph grows slower than in the $q=0$ and $q=0.5$ random models. 

\vspace{-.1cm}

\begin{figure}[ht]
\[
\includegraphics[scale=0.7,trim = {5cm 1cm 5cm 0cm}]{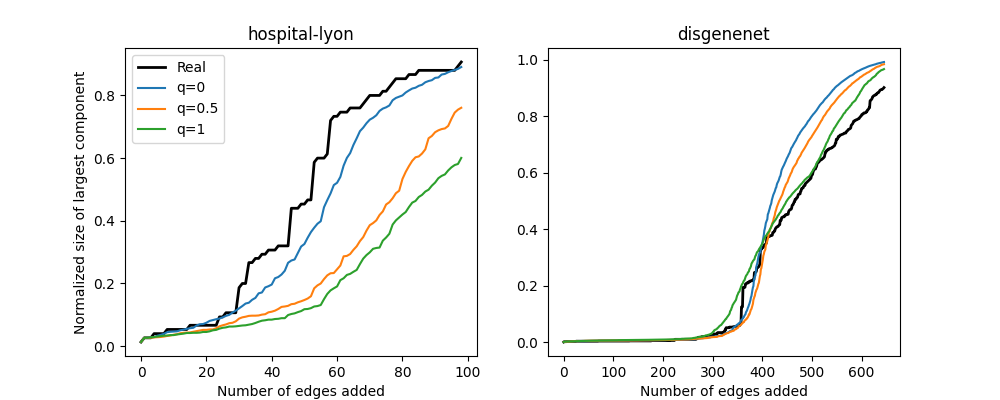}
\]
\caption{Giant component size vs.\ number of edges added in the adversarial growth process on the \textbf{hospital-lyon} graph (left) and the \textbf{disgenenet} graph (right). The curve is the point-wise average across 20 independent experiments: for the real graph the experiment is performed only once as the result will always be the same, and for the random models the graphs are resampled each time.}
\label{fig:exp2}
\end{figure}

\subsubsection*{Experiment 3: single-source diffusion}

This experiment, shown in Figure \ref{fig:exp3}, is perhaps the most substantial in showing the effect of simpliciality on a random process, namely, that information diffusion is slower on highly simplicial graphs vs. non-simplicial graphs. We note, however, that the diffusion process on \textbf{hospital-lyon} is still slower than that of a random model with $q=0.5$. Surely there are more features of this real graph not captured by random models that contribute to the slower diffusion time. 

\vspace{.35cm}

\begin{figure}[ht]
\[
\includegraphics[scale=0.75,trim = {5cm 1cm 5cm 1cm}]{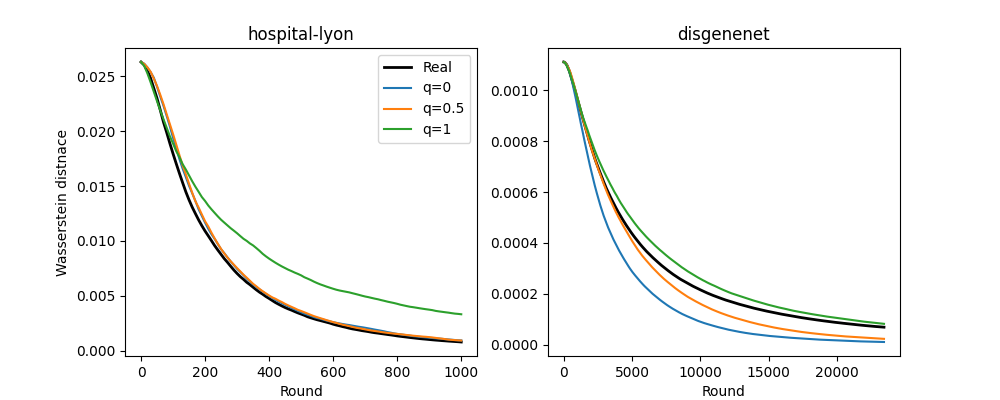}
\]
\caption{Wasserstein distance to uniform vs.\ number of rounds in the single-source diffusion process on the \textbf{hospital-lyon} graph (left) and the \textbf{disgenenet} graph (right). The curve is the point-wise average across 10000 independent experiments: for the real graph the chosen edges per round, as well as the location of the initial vertex with weight 1, are resampled each time, and for the random models the entire graphs are resampled each time.}
\label{fig:exp3}
\end{figure} 

\subsubsection*{Experiment 4: 10\% diffusion}

The result shown in Figure \ref{fig:exp4} mirrors the result in the previous experiment, except of course that the diffusion is much faster. 

\bigskip 

\begin{figure}[ht]
\[
\includegraphics[scale=0.75,trim = {5cm 1cm 5cm 1cm}]{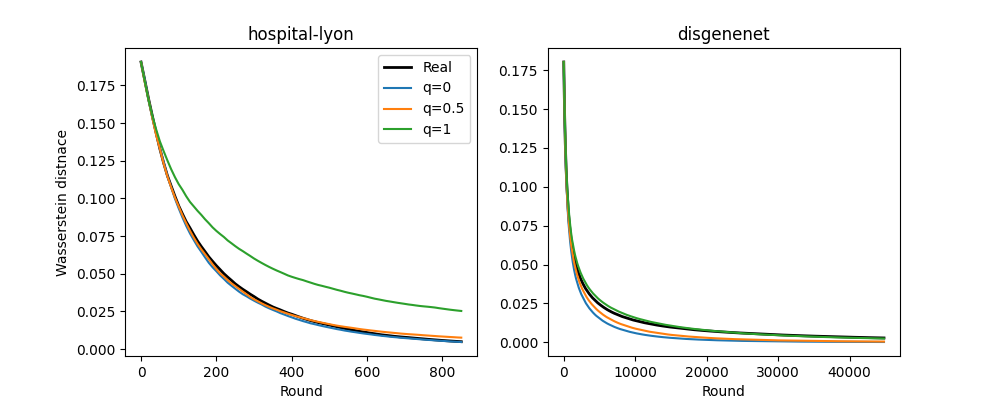}
\]
\caption{Wasserstein distance to uniform vs.\ number of rounds in the 10\% sprinkled diffusion process on the \textbf{hospital-lyon} graph (left) and the \textbf{disgenenet} graph (right). The curve is the point-wise average across 10000 independent experiments: for the real graph the chosen edges per round, as well as the location of the initial 10\% of vertices with weight 1, are resampled each time, and for the random models the entire graphs are resampled each time.}
\label{fig:exp4}
\end{figure}

\section{Conclusion}\label{sec:conclusion}

The phenomenon of edges inside of other edges is a feature of hypergraphs not present in graphs and, based on our results and on the preceding results of Landry, Young and Eikmeier, it is clear that this phenomenon is a key feature of real-world networks with multi-way interactions. The \measurename{} captures the strength of simplicial interactions in a graph and, from the collection of 10 real-world networks analysed, we have showed that (a) the \measurename{} is not at all consistent across the graphs, (b) the \measurename{} varies significantly even for graphs of a similar type (e.g., \textbf{contact high-school}, \textbf{contact primary-school}, and \textbf{hospital-lyon}), (c) the number of simplicial interactions involving edges of size $k, \ell > 2$ is not at all captured by the Chung Lu model, and (d) the \measurename{} can affect the outcome of random growth, adversarial growth, and information diffusion. We hope that our work continues to motivate research into the phenomenon of edges inside edges, and we discuss some potential follow ups to this research.

\subsection{Further research}\label{subsec:further research}

The \measurename{} involves the parameter $\E{\pairs{\hat{G}}}$ where $\hat{G} \sim \mathrm{CL}(G)$. Instead of approximating $\E{\pairs{\hat{G}}}$ as we do, one could compute $\E{\pairs{\hat{G}}}$ explicitly. For example, given a uniform degree sequence $\mathbf{d}$ and edge-size sequence $(m_{k_{\mathrm{min}}}, \dots, m_{k_{\mathrm{max}}})$, and conditioning on $\hat{G}$ containing no multiset edges, the probability that $e_1, e_2$ form a simplicial pair is
\[
{|e_2| \choose |e_1|} \Big/ {n \choose |e_1|} \,.
\]
Thus, by linearity of expectation, conditioning on the event that $\hat{G}$ has no multiset edges, we have 
\[
\E{\pairs{\hat{G}}} = \sum_{k=k_{\mathrm{min}}}^{k_{\mathrm{max}}-1} \sum_{\ell = k+1}^{k_{\mathrm{max}}} m_k m_\ell {\ell \choose k} \Big/ {n \choose k} \,.
\]
Thus, for a uniform degree sequence, $\E{\pairs{\hat{G}}}$ is relatively straightforward to compute. However, trying to compute this expectation if the degree sequence is not uniform is significantly harder. Finding a closed form for this expectation, or even a closed form approximation, would allow for a significantly faster algorithm for computing $\measure{G}$. Such a result would also allow for a better understanding of the nature of the simplicial matrix for both sparse and dense graphs. 

Understanding the degree to which edges form simplicial pairs could aid in predicting the composition of future edges, especially large edges, in temporal networks. If a graph has a high simplicial ratio, then a potential new edge should be given more weight based on the number of new simplicial pairs it creates, as well as on the size of the smaller edge in each pairs. For example, when considering the location for a new edge of size $5$ in a highly simplicial graph $G$, a location that creates many $(2,5)$ pairs should be given more weight, but perhaps a location that creates a single $(4,5)$ pair should be given \textit{even more} weight. In any case, incorporating simpliciality in the link prediction problem should improve existing algorithms, at least for highly simplicial graphs. 

Along with the \measurename{} and \matrixname{}, we introduce temporal variants. In our experiments where only the first instance of an edge is kept in a temporal network, we find that, typically, more bottom-up pairs are generated than top-down pairs, in part because there are more small multi-edges than large multi-edges. There are of course other ways to measure the difference in frequency between bottom-up pairs and top-down pairs. For example, we could insist that a simplicial pair $e_k, e_\ell$ is ``temporally relevant'' if and only if both $e_k$ and $e_\ell$ were born within the same $\epsilon$-window of time. In this case, we could measure the frequency of $e_k$ pairs followed shortly by $e_\ell$ pairs, and vice versa. The temporal formation of simplicial pairs could once again be valuable for the task of link prediction. 

\clearpage

\bibliography{biblio}

\appendix


\section{All experiments} \label{app: experiments}

Here we show the results of the random growth, single-source diffusion, and 10\% diffusion experiments. Due to the time complexity of computing edge-betweenness, we are unable to perform the adversarial growth experiment for all 10 graphs. Note that \textbf{ubuntu (edge-chopped)} is the subgraph of \textbf{tags-ask-ubuntu} containing only the first 20000 edges. The simplicial ratio of this edge-chopped graph is $\approx 0.37$ and so this subgraph is even less simplicial than the whole graph.  

The experiments are presented in the the following order: random growth, single-source diffusion, and 10\% diffusion. Each of the three figures are presented on two separate pages. 



\newpage
\[
\includegraphics[scale=0.7,trim = {5cm 1cm 5cm 1cm}]{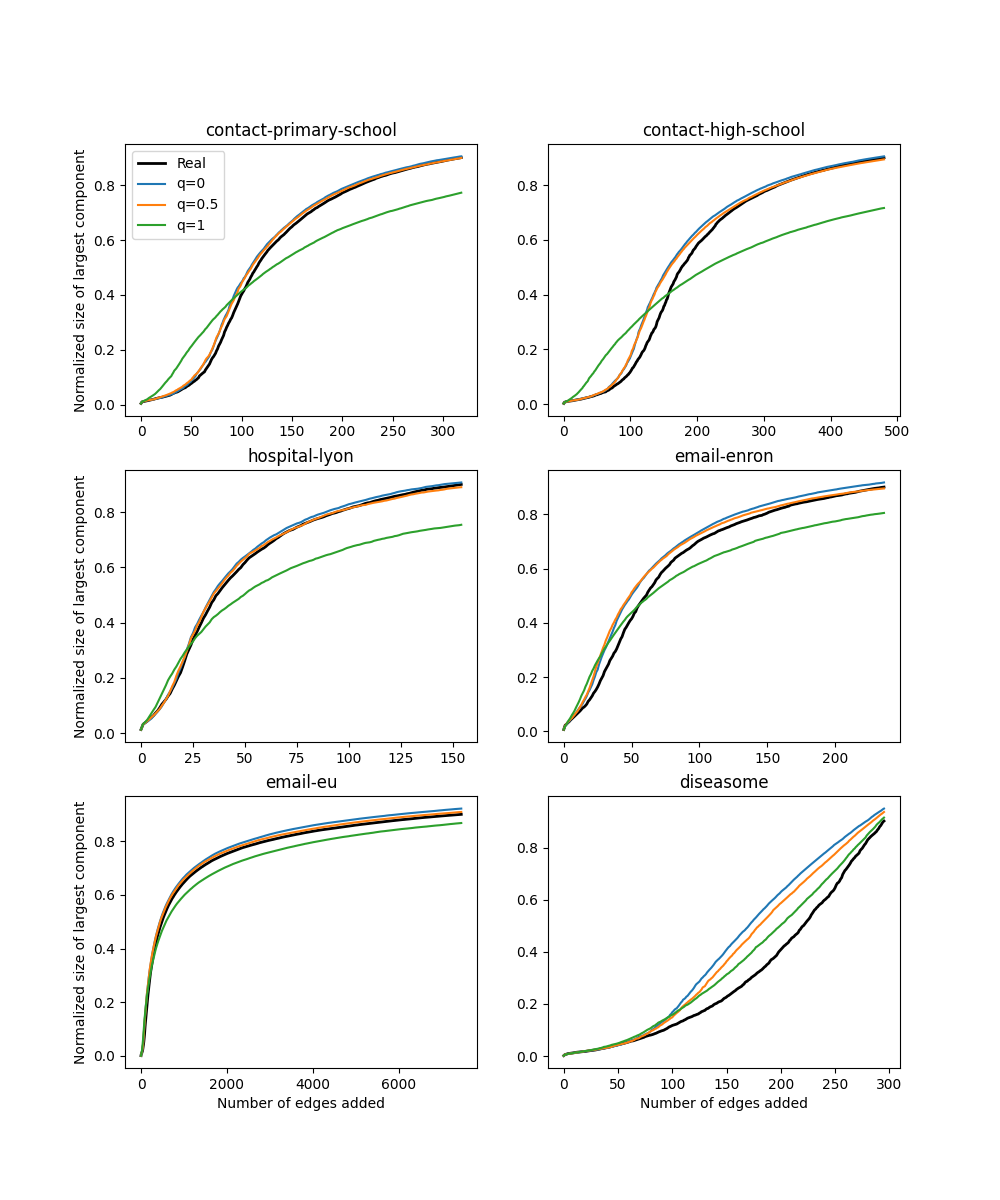}
\]
\newpage
\begin{figure}[H]
\[
\includegraphics[scale=0.7,trim = {5cm 1cm 5cm 1cm}]{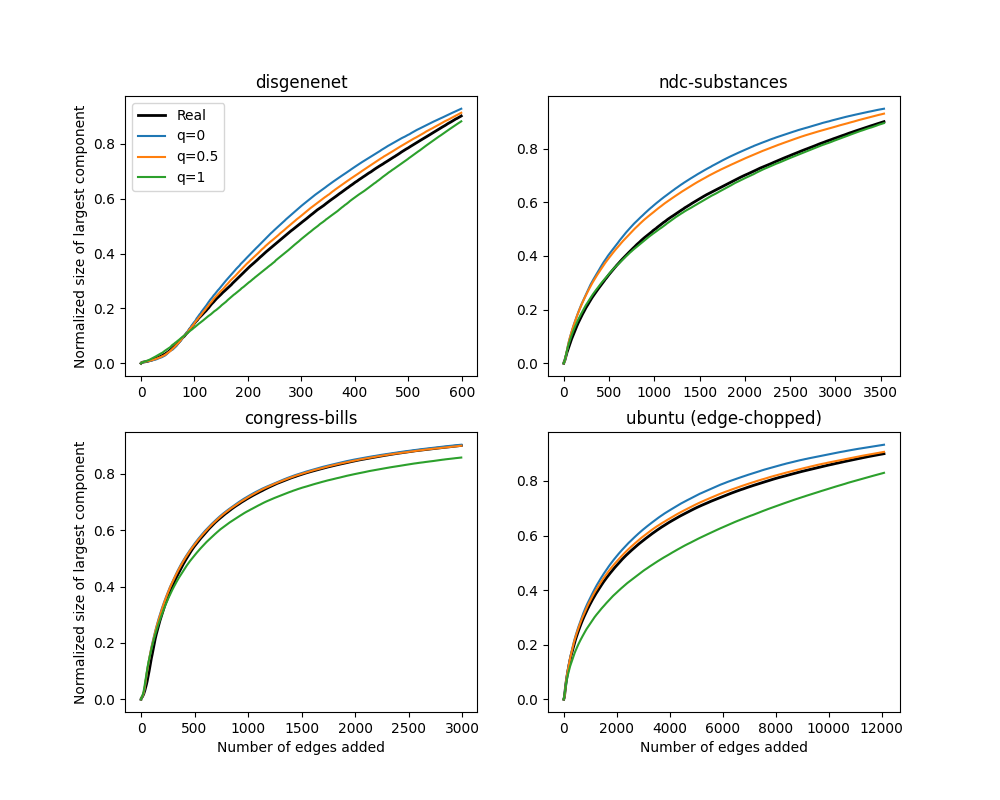}
\]
\caption{Giant component size (normalized by the number of vertices) vs.\ number of edges added in the random growth process for all 10 graphs. The curve is the point-wise average across 10000 independent experiments: for the real graph the edges are resampled each time, and for the random models the entire graphs are resampled each time.}
\label{fig:all random growth}
\end{figure}

\newpage



\[
\includegraphics[scale=0.7,trim = {5cm 1cm 5cm 1cm}]{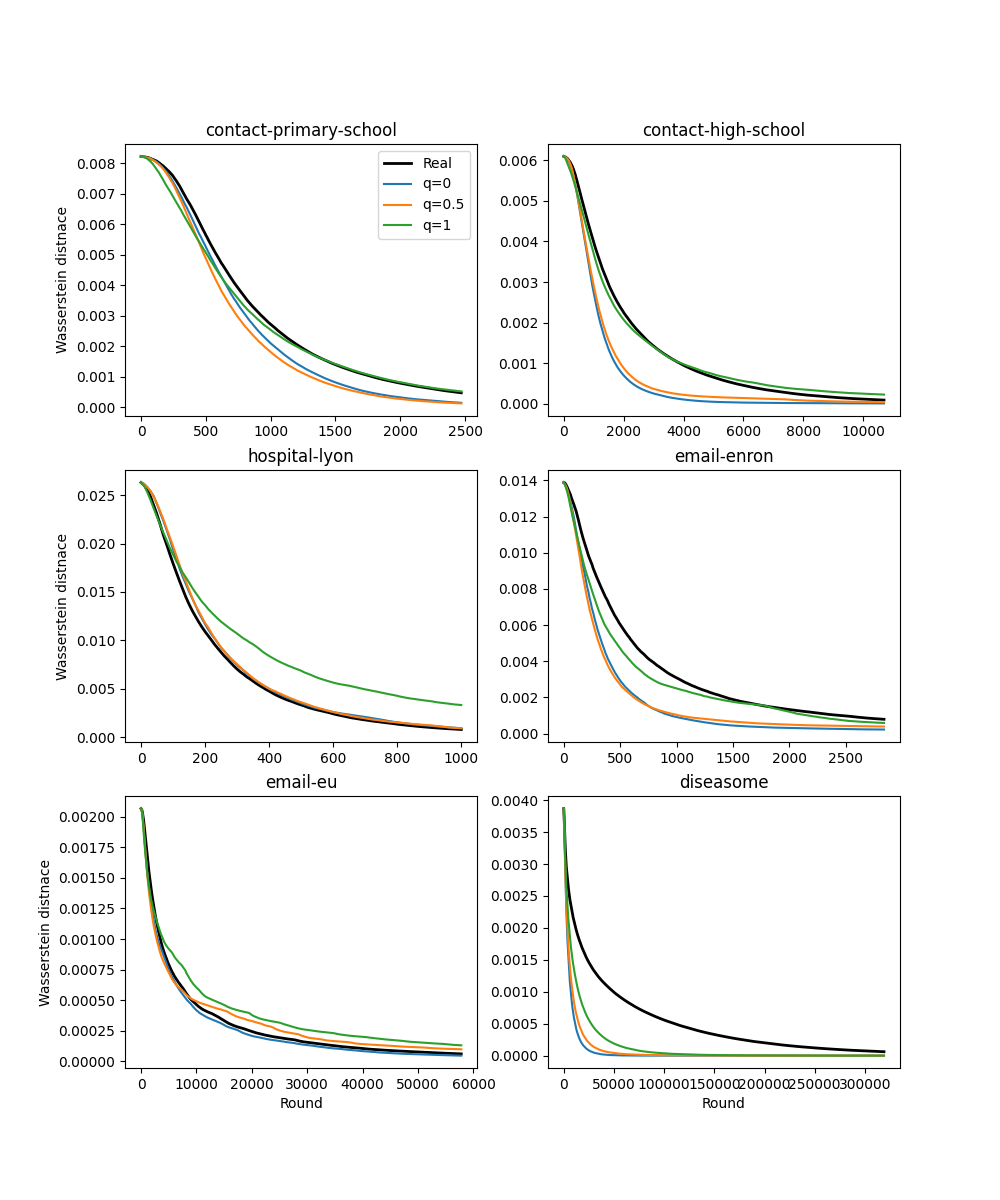}
\]
\newpage

\begin{figure}[H]
\[
\includegraphics[scale=0.7,trim = {5cm 1cm 5cm 1cm}]{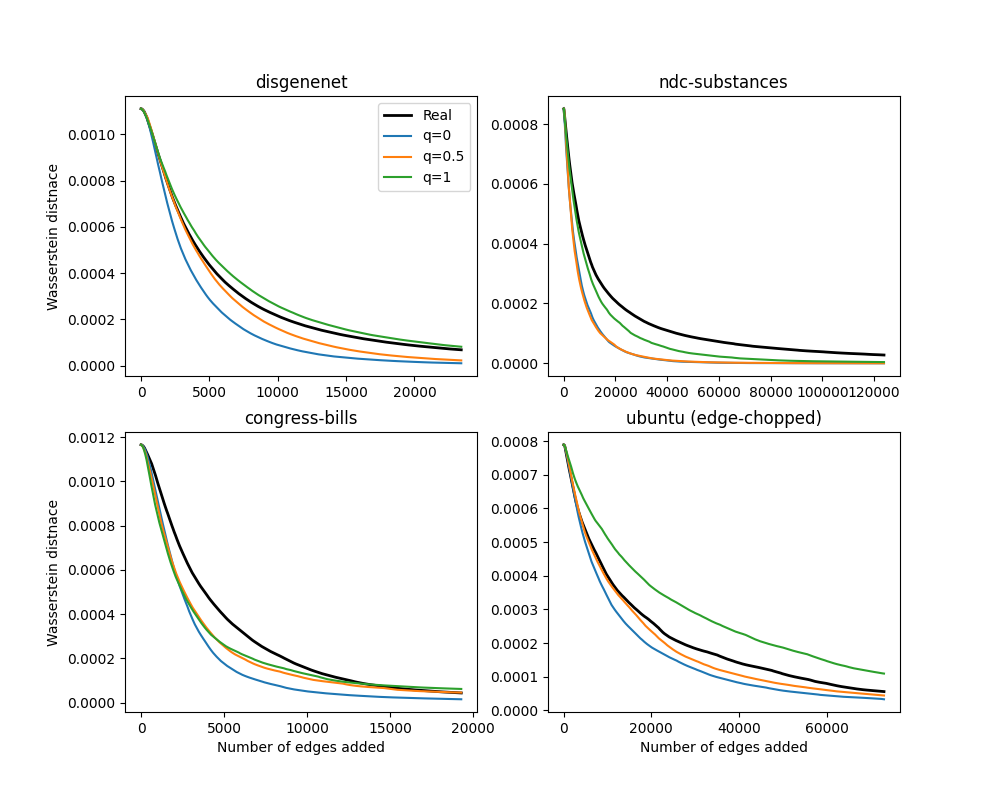}
\]
\caption{Wasserstein distance to uniform vs.\ number of rounds in the single-source diffusion process for all 10 graphs. The curve is the point-wise average across 10000 independet experiments: for the real graph the chosen edges per round, as well as the location of the initial vertex with weight 1, are resampled each time, and for the random models the entire graphs are resampled each time.}
\label{fig:all ss diffusion}
\end{figure}

\newpage

\[
\includegraphics[scale=0.7,trim = {5cm 1cm 5cm 1cm}]{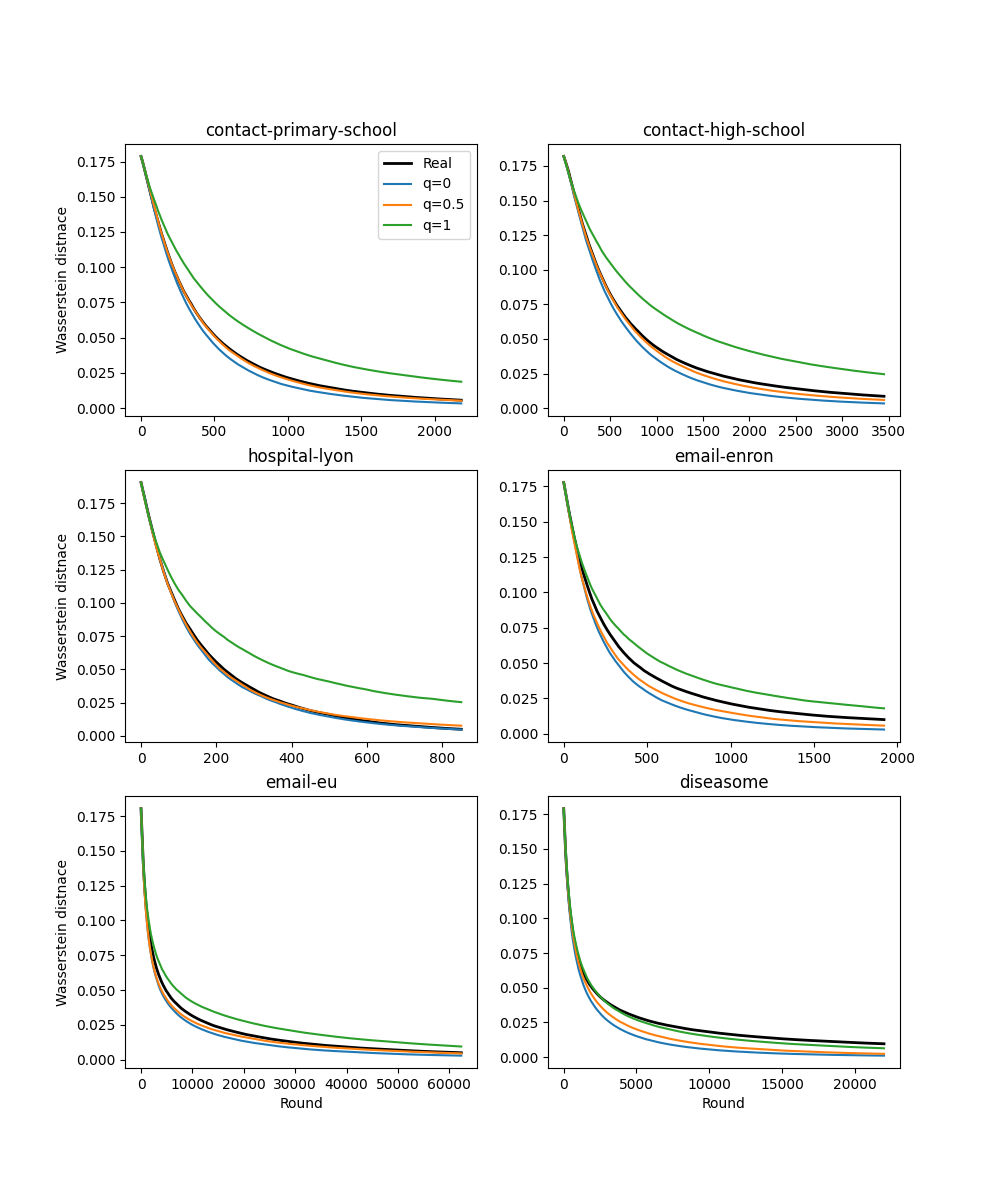}
\]
\newpage
\begin{figure}[H]
\[
\includegraphics[scale=0.7,trim = {5cm 1cm 5cm 1cm}]{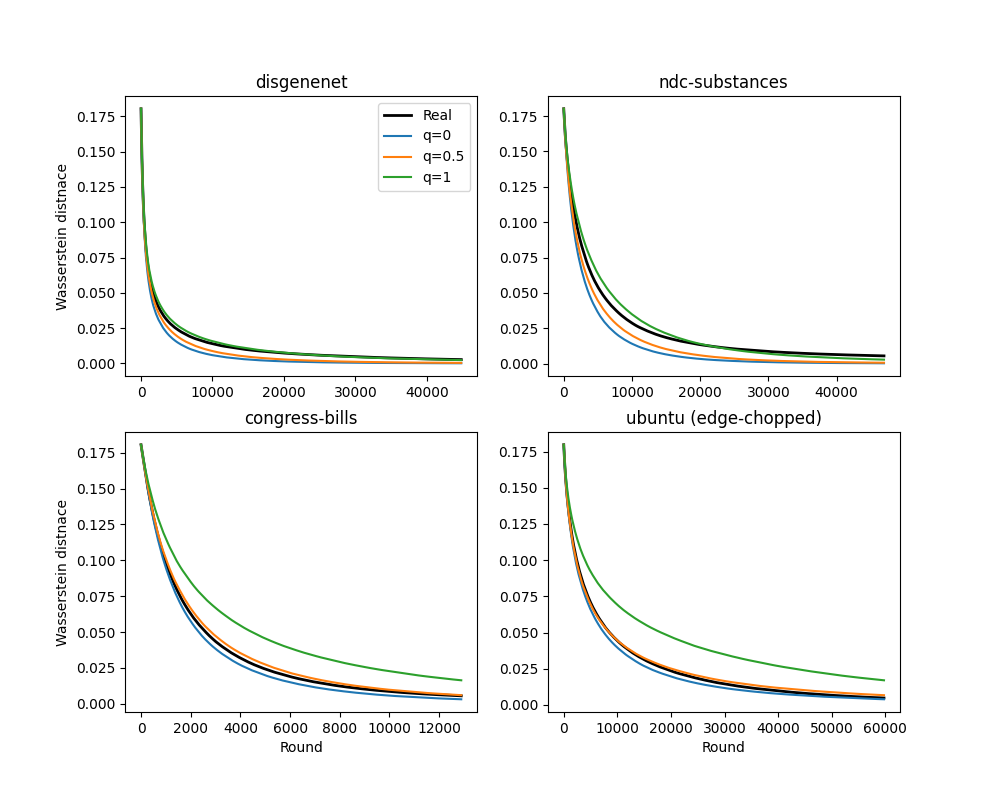}
\]
\caption{Wasserstein distance to uniform vs.\ number of rounds in the 10\% sprinkled diffusion process for all 10 graphs. The curve is the point-wise average across 10000 independent experiments: for the real graph the chosen edges per round, as well as the location of the initial 10\% of vertices with weight 1, are resampled each time, and for the random models the entire graphs are resampled each time.}
\label{fig:all 10p diffusion}
\end{figure}

\section{Algorithms} \label{app: algs}

\subsection{Estimating the expected number of simplicial pairs}\label{subapp: approx alg}

To compute the simplicial ratio of a graph $G$, we must first compute the expected number of simplicial pairs in $\hat{G} \sim \mathrm{CL}(G)$. As discussed in Section~\ref{sec:conclusion}, computing this expectation is quite difficult. In this section, we outline a Monte Carlo approximate technique for this expectation.

For a degree sequence $\mathbf{d} = (d_1,\dots d_n)$ and an edge size $k$, write $\Cprob{\text{simple}}{\mathbf{d}, k}$ for the probability that {Algorithm~\ref{alg:chung lu edge}} generates a simple edge when given inputs $\mathbf{d}$ and $k$. For a graph $G$ with degree sequence $\mathbf{d}$, we first approximate $\Cprob{\text{simple}}{\mathbf{d}, k}$ for all edge sizes $k$ in $E(G)$. To do this, we chose a number of samples $s$, sample $s$ edges independently as Algorithm~\ref{alg:chung lu edge} with $(\mathbf{d}, k)$, and compute the ratio $x/s$ where $x$ is the number of simple edges generated. 
In all experiments performed for the paper, we use $s=1000$.

With $\Cprob{\mathrm{simple}}{\mathbf{d}, k}$ approximated for all edge sizes $k$, we can now approximate the number of simplicial pairs. We will show the algorithm for computing the expected number of $(3, 5)$-pairs here, as the generalization is straightforward to interpret but difficult to notate. Write $|\mathbf{d}| := \sum_{i \in [n]} d_i$. For an edge $e = \{v_1, \dots, v_5\}$, the probability that an edge $e'$ of size $3$ generated by {Algorithm~\ref{alg:chung lu edge}} is (a) simple and (b) satisfies $e' \subset e$ is given by
\begin{equation} \label{eq: 3 5 pair}
\sum_{1 \leq a<b<c \leq 5} \frac{ 3! \, d_{v_a} d_{v_b} d_{v_c}}{ (|\mathbf{d}|)^3 \, \Cprob{\mathrm{simple}}{\mathbf{d}, 3}} \,.
\end{equation}
To break this down, consider only the probability that $e' = \{v_1, v_2, v_3\}$. {Algorithm~\ref{alg:chung lu edge}} can generate this edge in $3!$ different orders, and the probability of generating the edge in each case is
\[
\frac{d_{v_1} d_{v_2} d_{v_3}}{|\mathbf{d}|^3} \,.
\]
It can also happen that {Algorithm~\ref{alg:chung lu edge}} generates a multi-edge, requiring us to sample again. Thus, the probability of \textit{eventually} sampling the edge $e' = \{v_1, v_2, v_3\}$ is 
\begin{align*}
\sum_{i \geq 0} \left( 1 - \Cprob{\mathrm{simple}}{\mathbf{d}, 3} \right)^i \frac{3! \, d_{v_1} d_{v_2} d_{v_3}}{|\mathbf{d}|^3} 
&=
\frac{3! \, d_{v_1} d_{v_2} d_{v_3}}{|\mathbf{d}|^3}  \sum_{i \geq 0} \left( 1 - \Cprob{\mathrm{simple}}{\mathbf{d}, 3} \right)^i \\
&=
\frac{3! \, d_{v_1} d_{v_2} d_{v_3}}{|\mathbf{d}|^3}  \left(\frac{1}{1-(1-\Cprob{\mathrm{simple}}{\mathbf{d}, 3})}\right)\\
&=
\frac{3! \, d_{v_1} d_{v_2} d_{v_3}}{(|\mathbf{d}|^3) \, \Cprob{\mathrm{simple}}{\mathbf{d}, 3}} \,.
\end{align*}
Summing over all ${5 \choose 3}$ possible $3$-edges inside $e$ gives us~$(\ref{eq: 3 5 pair})$. 

We now approximate the number of $(k, \ell)$ simplicial pairs as follows. 
\begin{enumerate}
\item Choose some sampling number $s$. Then, sample $s$ independent edges via {Algorithm~\ref{alg:chung lu edge}}with $(\mathbf{d}, \ell)$.
\item For each edge, compute the probability of generating a $(k, \ell)$ simplicial pair.
\item Compute the average and multiply this result by $m_k m_l$, where $m_k$ is the number of edges of size $k$, and similarly for $m_\ell$. 
\end{enumerate}

As mentioned previously, for all of the experiments performed in this paper, we chose $s = 1000$.

\subsection{Constructing a connected skeleton of a random graph}\label{subapp: skeleton}

We will generate a connected skeleton for our random graph via multiplicative coalescence. In short, multiplicative coalescence is a process in which particles in a space join together at a rate proportional to the product of their masses. We point the reader to \cite{YeoMULTIPLICATIVEC} for an overview on the multiplicative coalescence process. In the context of generating random graphs, multiplicative coalescence is the process where new edges joining disjoint components are chosen with probability proportional to the product of the weights of the components.

We will describe {Algorithm~\ref{alg:skeleton}} in words before presenting it as pseudo-code. Let $\mathbf{d} := (d_1, \dots, d_n)$ be a degree sequence and $\mathbf{m} := (m_{k_{\mathrm{min}}}, \dots, m_{k_{\mathrm{max}}})$ be a sequence of edge sizes. We construct the skeleton of our graph as follows. 
\begin{enumerate}
\item Initially, we have an empty edge list $E = \{\}$ and a collection of components, one for each vertex. For component $C = \{v\}$, define the weight of $C$, written $w(C)$, as $w(C) := d_v$. 
\item We generate a random edge-size list $S$ as per {Algorithm~\ref{alg:simplicial}}, i.e., a uniform permutation containing $m_k$ copies of $k$ for each edge size $k$.
\item Iteratively until the graph is connected, we do the following.
\begin{enumerate}
\item Choose a size $k$ from $S$ (iteratively). 
\item Sample $k$ components independently, each component $C$ being chosen with probability proportional to $w(C)$. If the chosen components $C_1, \dots, C_k$ are not all unique, discard them all and sample again (repeating until we have a collection of distinct components).
\item For each component $C$ chosen in the previous step, randomly sample a designated vertex for $C$; for $v \in C$, choose $v$ as the designated vertex for $C$ with probability $d_v/\sum_{u \in C} d_u$. 
\item Construct the edge $e$ consisting of all the designated vertices. Add $e$ to $E$, remove the chosen components $C_1, \dots, C_k$, and create a new component $C = \cup_{j \in [k]} C_{k}$ with $w(C) = \sum_{i \in [k]} C_i$. 
\end{enumerate} 
\end{enumerate}
If, just before the graph is fully connected, the chosen size $k$ is greater than the number of components $c$, we generate the last edge of the connected skeleton by connecting the final $c$ components as per step 3 (with $k$ replaced by $c$) and sampling the remaining $k-c$ vertices as per the usual Chung-Lu sampling technique, i.e., using {Algorithm~\ref{alg:chung lu edge}}. We note that, other than potentially the last edge constructed, an edge constructed in step~3 is equivalent to an edge generated by {Algorithm~\ref{alg:chung lu edge}} conditioned on this edge joining $k$ distinct components. We use this observation to simplify {Algorithm~\ref{alg:skeleton}}. We will simplify {Algorithm~\ref{alg:skeleton}} by writing ``update [collection of components]'' after generating an edge. 

\begin{algorithm}[ht]
\caption{Connected skeleton. }\label{alg:skeleton}
\begin{algorithmic}[1]
\Require $(d_1,\dots,d_n)$, $(m_{k_{\mathrm{min}}}, \dots, m_{k_{\mathrm{max}}})$
\State Initialize edge list $E = \{\}$, a random edge-size list $S$ as per {Algorithm~\ref{alg:simplicial}}, and a collection of components $\cC = \big\{ C_v := \{v\} \big| v \in [n] \big\}$.
\For{$k \in S$}
\If{$k \leq |\cC|$}
\Repeat
\State Sample $e \sim \mathbf{Algorithm~\ref{alg:chung lu edge}} \Big( (d_1,\dots,d_n), k \Big)$.
\Until{$\big|e \cap C \big| \leq 1$ for all $C \in \cC$}
\State Set $E = E \cup e$ and update $\cC$.
\Else
\State Set $c = |\cC|$.
\Repeat
\State Sample $e' \sim \mathbf{Algorithm~\ref{alg:chung lu edge}} \Big( (d_1,\dots,d_n), c \Big)$.
\Until{$\big|e' \cap C \big| \leq 1$ for all $C \in \cC$} 
\State Sample $e'' \sim \mathbf{Algorithm~\ref{alg:chung lu edge}} \Big( (d_1,\dots,d_n), k-c \Big)$.
\State Set $E = E \cup \{e' \cup e''\}$ and update $\cC$.
\EndIf
\If{$|\cC| = 1$}
\State Return $E$
\EndIf
\EndFor
\State Return $E$
\end{algorithmic}
\end{algorithm}

Once we generate a connected skeleton via {Algorithm~\ref{alg:skeleton}}, we then update the parameter $(m_{k_{\mathrm{min}}}, \dots, m_{k_{\mathrm{max}}})$ (by subtracting, from $m_k$, the number of edges of size $k$ that were generated for each $k$) and generate the rest of the simplicial Chung-Lu graph via {Algorithm~\ref{alg:simplicial}} with updated parameter $(m_{k_{\mathrm{min}}}, \dots, m_{k_{\mathrm{max}}})$ and initial (non-empty) edge list $E$.

\end{document}